%% file: main-arxiv.tex
\documentclass[sigconf,noacm]{acmart}
\AtBeginDocument{%
  }

\setcopyright{acmlicensed}
\copyrightyear{2018}
\acmYear{2018}
\acmDOI{XXXXXXX.XXXXXXX}
\acmConference[Conference acronym 'XX]{Make sure to enter the correct
  conference title from your rights confirmation email}{June 03--05,
  2018}{Woodstock, NY}
\acmISBN{978-1-4503-XXXX-X/2018/06}

\input{packages}

\input{macros}

\usepackage{amsfonts}

\usepackage{amsthm}

\usepackage{soul,xcolor}

\usepackage{enumitem}

\sethlcolor{red}

\newcommand*\rcircled[1]{\tikz[baseline=(char.base)]{
    \node[fill=black,text=white, shape=circle,draw=black,inner sep=.6pt] (char) {#1};}}

    \newcommand{\repeattheorem}[1]{%
  \begingroup
  \renewcommand{\thetheorem}{\ref{#1}}%
  \expandafter\expandafter\expandafter\theorem
  \csname reptheorem@#1\endcsname
  \endtheorem
  \endgroup
}
\newcommand*\emptycirc[1][1ex]{\tikz\draw (0,0) circle (#1);} 
\newcommand*\halfcirc[1][1ex]{%
  \begin{tikzpicture}
  \draw[fill] (0,0)-- (90:#1) arc (90:270:#1) -- cycle ;
  \draw (0,0) circle (#1);
  \end{tikzpicture}}
\newcommand*\fullcirc[1][1ex]{\tikz\fill (0,0) circle (#1);} 

\NewEnviron{reptheorem}[1]{%
  \global\expandafter\xdef\csname reptheorem@#1\endcsname{%
    \unexpanded\expandafter{\BODY}%
  }%
  \expandafter\theorem\BODY\unskip\label{#1}\endtheorem
}

\renewcommand\footnotetextcopyrightpermission[1]{}
\begin{document}

\title{DTL: Data Tumbling Layer \\ A Composable Unlinkability for Smart Contracts}

\author{Mohsen Minaei}
\affiliation{%
  \institution{Visa Research}
  \country{}
}

\author{Pedro Moreno-Sanchez}
\affiliation{%
  \institution{IMDEA Software Institute, Visa Research, MPI-SP}
  \country{}
}

\author{Zhiyong Fang}
\affiliation{%
  \institution{Texas A\&M University}
  \country{}
}

\author{Srinivasan Raghuraman}
\affiliation{%
  \institution{Visa Research and MIT}
  \country{}
}

\author{Navid Alamati}
\affiliation{%
  \institution{Visa Research}
  \country{}
}

\author{Panagiotis Chatzigiannis}
\affiliation{%
  \institution{Visa Research}
  \country{}
}

\author{Ranjit Kumaresan}
\affiliation{%
  \institution{Visa Research}
  \country{}
}

\author{Duc V. Le}
\affiliation{%
  \institution{Visa Research}
  \country{}
}








\renewcommand{\shortauthors}{Trovato et al.}

\input{abstract}





\maketitle

\input{introduction}

\input{relwork}


\input{preliminaries}

\input{problem-statement}
\input{construction}

\input{overview}
\input{applications}

\input{evaluation}

\bibliographystyle{ACM-Reference-Format}
\bibliography{abbrev3,references,crypto}

\appendix

\input{tmp}

\input{security-analysis}

\end{document}

%% file: packages.tex
\usepackage{hyperref}
\usepackage{xcolor}
\usepackage[commandnameprefix=ifneeded,commentmarkup=uwave]{changes}
\usepackage{url}
\usepackage[n,landau,notions,ff,mm]{cryptocode}
\usepackage{xspace}
\usepackage{placeins} 
\usepackage[framemethod=TikZ]{mdframed}
\usepackage{dashbox}
\usepackage[colorinlistoftodos]{todonotes}
\setuptodonotes{inline}
\usepackage{booktabs}
\usepackage{paralist}
\usepackage[noabbrev,capitalize]{cleveref}
\usepackage{multirow}
\usepackage{makecell}
\usepackage{colortbl}
\usepackage{xcolor} 
\usepackage{tikz}
\usetikzlibrary{shapes, positioning, arrows.meta}

\usepackage{pgf-umlsd}
\newcommand{\Mess}[4][0]{
  \stepcounter{seqlevel}
  \path
  (#2)+(0,-\theseqlevel*\unitfactor-0.7*\unitfactor) node (mess from) {};
  \addtocounter{seqlevel}{#1}
  \path
  (#4)+(0,-\theseqlevel*\unitfactor-0.7*\unitfactor) node (mess to) {};
  \draw[->,>=angle 60] (mess from) -- (mess to) node[midway, above]
  {\footnotesize #3};
}
\newcommand{\IMess}[4][0]{
  \stepcounter{seqlevel}
  \path
  (#2)+(0,-\theseqlevel*\unitfactor-0.7*\unitfactor) node (mess from) {};
  \addtocounter{seqlevel}{#1}
  \path
  (#4)+(0,-\theseqlevel*\unitfactor-0.7*\unitfactor) node (mess to) {};
  \path[->,>=angle 60] (mess from) -- (mess to) node[midway, above]
  {\footnotesize #3};
}

%% file: macros.tex
\definecolor{cblue}{rgb}{0.0, 0.28, 0.9}   
\definechangesauthor[name=duc, color=cblue]{duc}
\definechangesauthor[name=zhiyong, color=blue]{zhiyong}
\definecolor{pms}{rgb}{0.0, 0.28, 0.4}   
\definechangesauthor[name=pedro, color=pms]{pedro}
\definechangesauthor[name=mohsen, color=cyan]{mohsen}
\definechangesauthor[name=ranjit, color=red]{ranjit}

\newcommand{\glcomment}[1]{\texttt{\scriptsize \color{gray}{// #1}}}

\newcommand{\fun}[1]{\ensuremath{\mathtt{#1}}\xspace}
\newcommand{\key}[1]{\ensuremath{\mathsf{#1}}\xspace}
\newcommand{\pparagraph}[1]{\noindent \textbf{#1.}}

\newcommand{\bset}{\{0,1\}}

\newcommand{\FF}{\mathbb{F}}
\newcommand{\NN}{\mathbb{N}}

\newcommand{\poly}{\mathsf{poly}}
\newcommand{\negl}{\mathsf{negl}}
\newcommand{\adv}{\ensuremath{\mathcal{A}}}
\newcommand{\bdv}{\ensuremath{\mathcal{B}}}
\newcommand{\ddv}{\ensuremath{\mathcal{D}}}
\newcommand{\sdv}{\ensuremath{\mathcal{S}}}
\newcommand{\oracle}{\mathcal{O}}
\newcommand{\support}{\mathsf{SUPP}}

\newtheorem{theorem}{Theorem}
\newtheorem{claim}{Claim}

\newtheorem{definition}{Definition}
\newtheorem{lemma}{Lemma}

\newcommand{\rt}{\ensuremath{\mathsf{root}}\xspace}
\newcommand{\initialize}{\fun{Init}}
\newcommand{\mkverify}{\fun{Verify}}
\newcommand{\prove}{\fun{Prove}}

\newcommand{\zksnark}{zk-SNARK\xspace}
\newcommand{\zkschemeinitial}{\ensuremath{\Pi}\xspace}
\newcommand{\crs}{\key{crs}\xspace}

\newcommand{\zkprove}{\ensuremath{\zkschemeinitial.\fun{Prove}}}
\newcommand{\zkverify}{\ensuremath{\zkschemeinitial.\fun{Verify}}}
\newcommand{\statement}{\key{stmt}}

\newcommand{\zkproof}{\ensuremath{\pi_{\mathsf{zk}}}\xspace}
\newcommand{\CrsGen}{\mathsf{CrsGen}}
\newcommand{\CrsSim}{\mathsf{CrsSim}}
\newcommand{\PrvSim}{\mathsf{PrvSim}}

\newcommand{\tagscheme}{\ensuremath{\mathsf{TAG}}\xspace}
\newcommand{\tagpp}{\ensuremath{\pp_{\tagscheme}}\xspace}
\newcommand{\tagsetup}{\ensuremath{\mathtt{TagSetup}}\xspace}
\newcommand{\tagkeygen}{\ensuremath{\mathtt{TagKGen}}\xspace}
\newcommand{\tageval}{\ensuremath{\mathtt{TagEval}}\xspace}
\newcommand{\tagOnewayExp}{\ensuremath{\mathsf{OneWay}_{\tagscheme, \adv}}\xspace}
\newcommand{\tagPRExp}{\ensuremath{\mathsf{PR}_{\tagscheme, \adv}}\xspace}
\newcommand{\nulli}{\ensuremath{\mathsf{tag}}\xspace}

\newcommand{\encsetup}{\ensuremath{\mathtt{EncSetup}}}
\newcommand{\enckeygen}{\ensuremath{\mathtt{EncKGen}}}
\newcommand{\encrypt}{\ensuremath{\mathtt{Encrypt}}}
\newcommand{\decrypt}{\ensuremath{\mathtt{Decrypt}}}
\newcommand{\encpp}{\ensuremath{\pp_{enc}}}

\newcommand{\DDM}{\ensuremath{\mathtt{DTL}\xspace}}
\newcommand{\fixed}{\ensuremath{\mathsf{fixed}\xspace}}
\newcommand{\arbit}{\ensuremath{\mathsf{arb}\xspace}}
\newcommand{\fDDM}{\ensuremath{\DDM_\fixed\xspace}}
\newcommand{\aDDM}{\ensuremath{\DDM_\arbit\xspace}}
\newcommand{\dcmsetup}{\fun{Setup}}
\newcommand{\dcmcreatecoin}{\fun{Create}}
\newcommand{\dcmredeemcoin}{\fun{Redeem}}
\newcommand{\dcmaccumulatecoin}{\fun{Accumulate}}

\newcommand{\dcmverify}{\fun{Verify}}

\newcommand{\dcmcoinpk}{\key{cpk}\xspace}
\newcommand{\dcmcoinsk}{\key{csk}}
\newcommand{\dcmcoinaccstate}{\key{st}}

\newcommand{\dcmtag}{\key{tag}}
\newcommand{\dcmproof}{\ensuremath{\pi}}

\newcommand{\seqx}{\ensuremath{(x_i)_{i=1}^n}\xspace}

\newcommand{\dcmncoins}{\ensuremath{(\dcmcoinpk_i)_{i=1}^n}}

\newcommand{\dcmexpds}{\mathsf{ExpOneMoreRedeem}}
\newcommand{\dcmexptheft}{\mathsf{ExpTheft}}
\newcommand{\dcmexpunlink}{\mathsf{ExpUnlink}}

\newcommand{\dcmcoinlist}{C}
\newcommand{\dcmredeemlist}{R}

\newcommand{\dcmdata}{\key{data}}

\newcommand{\dcmnonslander}{\mathsf{ExpNSlander}}
\newcommand{\drootlist}{\mathsf{AccList}_{wdr}}
\newcommand{\dpl}{\ensuremath{\mathsf{AccHistory}}\xspace}

\newcommand{\wnl}{\ensuremath{\mathsf{TagList}}\xspace}

\newcommand{\acceptdeposit}{\fun{AcceptDeposit\xspace}}
\newcommand{\issueconfiwithdraw}{\fun{UCWithdraw\xspace}}
\newcommand{\issuewithdraw}{\fun{IssueWithdraw\xspace}}
\newcommand{\acceptconfideposit}{\fun{ConfidentialDeposit\xspace}}

\newcommand{\setup}{\fun{Setup}}
\newcommand{\witness}{\ensuremath{\mathsf{wit}}\xspace}
\newcommand{\cm}{\mathsf{cm}}
\newcommand{\commit}{\ensuremath{\mathsf{COM}\xspace}}
\newcommand{\commitcreate}{\fun{Commit}\xspace}

\newcommand{\ek}{\key{ek}\xspace}
\newcommand{\dk}{\key{dk}\xspace}
\newcommand{\bal}{\ensuremath{\key{bal}\xspace}}
\newcommand{\pk}{\key{prvK}\xspace}
\newcommand{\vk}{\key{vrfyK}\xspace}

\newcommand{\sparam}{\ensuremath{1^\lambda}\xspace}
\newcommand{\define}{\ensuremath{:=}\xspace}
\newcommand{\bit}{b\xspace}

\newcommand{\CreateDTX}{\fun{CreateDepositTx}\xspace}
\newcommand{\CreateCFWTX}{\fun{CreateUCWithdrawingTx}\xspace}
\newcommand{\CreateWTX}{\fun{CreateWithdrawTx}\xspace}
\newcommand{\CreateCFTX}{\fun{CreateCDepositTx}\xspace}

\newcommand{\pp}{\key{pp}\xspace}

\newcommand{\sample}{\ensuremath{\xleftarrow{\$}}\xspace}

\newcommand{\xcoin}{\ensuremath{\mathsf{amt}}\xspace}

\newcommand{\amt}{\ensuremath{\mathsf{amt}}\xspace}

\newcommand{\com}{\mathsf{com}\xspace}

\newcommand{\addr}{\ensuremath{\mathsf{addr}}}
\newcommand{\cdtl}{\ensuremath{\mathsf{C}_\mathsf{DTL}}\xspace}
\newcommand{\ccon}{\ensuremath{\mathsf{C}_\mathsf{Conf}}\xspace}

\newcommand{\tx}{\textsf{tx}\xspace}
\newcommand{\txd}{\textsf{tx}_{\mathsf{Deposit}}\xspace}
\newcommand{\txcd}{\textsf{tx}_{\mathsf{CDeposit}}\xspace}

\newcommand{\txw}{\textsf{tx}_{\mathsf{withdraw}}\xspace}
\newcommand{\txcw}{\textsf{tx}_{\mathsf{UCWithdraw}}\xspace}

\newcommand{\set}[1]{\ensuremath{\{#1\}}\xspace}

\newcommand{\tree}{\ensuremath{\mathsf{MT}\xspace}}

\newcommand{\adsproof}{\ensuremath{\mathsf{path}}}

\def \ifempty#1{\def\temp{#1} \ifx\temp\empty }



%% file: abstract.tex
\begin{abstract}
    We propose Data Tumbling Layer (DTL), a cryptographic scheme for non-interactive
    data tumbling. The core concept is to enable users to \emph{commit} to specific data and subsequently \emph{re-use} to the encrypted version of these data across different applications while removing the link to the previous data commit action.
    We define the following security and privacy notions for DTL: 
    \begin{inparaenum}[(i)]
        \item \emph{no one-more redemption}: a malicious user cannot redeem and
            use the same data more than the number of times they have committed the
            data; 
        \item \emph{theft prevention}: a malicious user cannot use data that has
            not been committed by them; 
        \item \emph{non-slanderabilty}: a malicious user cannot prevent an
            honest user from using their previously committed data; and
        \item \emph{unlinkability}: a malicious user cannot link tainted data
            from an honest user to the corresponding data after it has been tumbled. 
    \end{inparaenum}

    {To showcase the practicality of DTL, we use DTL to realize applications for}
    \begin{inparaenum}[(a)]
        \item
        unlinkable fixed-amount payments; 
        \item  
        unlinkable and confidential payments for variable amounts; 
        \item unlinkable {weighted} voting protocol.
    \end{inparaenum}
    {Finally, we implemented and evaluated all the proposed applications. 
    For the unlinkable and confidential payment application, a user can initiate such a transaction in less than $1.5s$ on a personal laptop. In
terms of on-chain verification, the gas cost is less than $1.8$ million.}

\end{abstract}

%% file: introduction.tex
\section{Introduction}
\label{sec:intro}

Distributed Ledger Technology is poised to revolutionize traditional banking by offering a secure, transparent, and efficient platform for financial transactions. Banks adopting distributed ledger technologies can significantly reduce operational costs and fraud risk, automate compliance, and provide real-time transaction processing. 
The robust and scalable nature of blockchain applications is further exemplified by the success of Ethereum, which supports more than $44$ million applications and oversees assets exceeding $224$ billion USD. 
This integration not only positions banks at the forefront of financial innovation, but also opens up new revenue avenues by catering to the burgeoning demand for decentralized financial services. 

However, distributed ledgers are commonly known for their inability to provide user privacy, a factor that hampers widespread adoption. Transactions on these networks publicly disclose details, such as transaction amounts and the addresses of both senders and recipients. Although this information does not directly reveal an individual's identity, the publicly accessible data can be subject to analysis,  potentially compromising a user's privacy.

In response to this weakness, both the academic and practitioner communities have proposed solutions that address various aspects of privacy, including unlinkability (e.g.,  AMR~\cite{le2020amr}), confidentiality (e.g.,  Zether~\cite{zether-bunz-2020}), and untraceability (e.g., Stealth Address~\cite{lai2019omniring}). 
These solutions offer a comprehensive method to protect user information. They can be broadly classified into \emph{(i) add-on privacy} solutions, which are built on existing non-private blockchains, and \emph{(ii) private-by-design blockchains}, like Monero~\cite{Noether2016-Ledger} and Zcash~\cite{sasson2014zerocash}.

However, these privacy-enhancing approaches lack \emph{composability} as they are specifically tailored for functions such as money transfers or obfuscating payment graphs, without the ability to integrate with other existing on-chain applications. Consequently, an important research question emerges: 
\begin{center}
\emph{How can we design privacy-enhancing solutions that can be composed with existing on-chain applications to enhance user privacy throughout the entire system?}
\end{center}


We believe that modularity and composability are fundamental
properties that are often overlooked by both practitioners and
academics in blockchain privacy solutions. While existing approaches
like mixers and privacy-preserving protocols tend to be tightly coupled
to specific applications (e.g., token transfers) or implement privacy
features from scratch, this work introduces a more systematic approach.
We extract unlinkability as a core primitive and develop it into a
composable building block - similar to how Tor provides a reusable
privacy layer for network communications. Our formalized primitive
can be seamlessly integrated into diverse smart contract applications
to provide unlinkability guarantees. By generalizing unlinkability
techniques to support both fixed and arbitrary data types, we enable
developers to add privacy features to complex smart contract logic
without reimplementing cryptographic protocols. This modular design
facilitates enhanced flexibility in building privacy-preserving systems
while maintaining strong security properties.

\pparagraph{Our Approach} Our aim is to provide a \emph{data tumbling layer} DTL, essentially an exchange of data among users, autonomously managed by a smart contract, \cdtl. 
We illustrate our approach in~\cref{fig:approach overview}. 
In more detail, assume a set of users, each of which wants to commit a piece of data, $\dcmdata_i \in \bset^{*}$ (e.g., $\dcmdata_i$ can be an integer representing a number of coins). In the first phase~\rcircled{1}, each user $u_i$ \emph{commits} the data $\dcmdata_i$ to the smart contract, $\cdtl$. In the second phase~\rcircled{2}, each user $u_j$ \emph{redeems} the (encrypted) version of data $\dcmdata_j$ from the smart contract for use in other applications.
Our proposed DTL addresses a number of technical challenges, namely \emph{security}, \emph{privacy}, and \emph{composability}. 

\begin{figure}[t]
    \centering
    \includegraphics[width=.6\columnwidth]{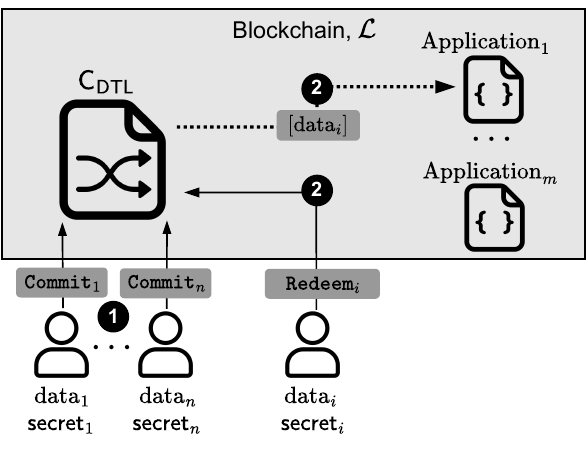}
    \vspace{-0.5cm}
    \caption{DTL Composability: Users leverage DTL to achieve unlinkability when interacting with other on-chain applications. Here, $[\dcmdata_i]$ denotes the encrypted version of $\dcmdata_i$.
    }
    \label{fig:approach overview}
\end{figure}

\subsection{Technical Challenges}

\pparagraph{Security} 
The correct execution of the code in a smart contract is verified by the distributed ledger operators (e.g., miners in Ethereum). However, this correctness guarantee does not suffice to provide smart contract users with guarantees about the security of their mixed data. For instance, the ownership of data $\dcmdata_i$ by user $u_i$ is often encoded as a cryptographic operation. Since a smart contract can, in principle, be invoked by any user, it opens the door for an adversary to steal the victim's data by breaching the cryptography used in the smart contract. In the data tumbling problem, the following security properties are of interest:
\begin{itemize}[leftmargin=0.5cm]
    \item \emph{Theft prevention.} An adversary must not be able to use data $\dcmdata_i$ previously committed by an honest user $u_i$. 
    \item \emph{Non-slanderability.} Assume that honest user $u_i$ commits data $\dcmdata_i$ into the smart contract. Then, the adversary must not be able to prevent the honest user from taking the data $\dcmdata_i$ to use it in other applications. 
\end{itemize}

\pparagraph{Privacy} 
The sensitivity of financial applications makes users demand strong privacy guarantees. The largely refuted hypothesis that pseudonyms (as used in Ethereum) provide a meaningful notion of privacy, has aroused scientific and industry interest for formally defined privacy notions and cryptographic protocols for blockchain applications with formally proven privacy guarantees. In particular, in the data tumbling problem, the following privacy notions are of interest: 
\begin{itemize}[leftmargin=0.5cm]
    \item \emph{Unlinkability for fixed data inputs.} Assume that honest users, $u_1, \ldots, u_n$, have committed their data, $\dcmdata_1, \ldots, \dcmdata_n$, to the smart contract where all $\dcmdata_i$ are the same (e.g., the same integer representing a transaction amount) and, in a second phase, these honest users have used $\dcmdata_1, \ldots, \dcmdata_n$ from the smart contract. In such a setting, an adversary with access to the information stored in the smart contract must not be able to link the committing and redeeming operations of user $u_i$. 
    \item \emph{Unlinkability for arbitrary data inputs.} For applications where the data being committed is unique (e.g., NFT), we want to offer the same guarantee as for the redeeming of fixed data inputs. In particular, an adversary cannot determine whether a specific data $\dcmdata_i$ committed previously is being redeemed.
\end{itemize}


\pparagraph{Composability} 
In distributed ledgers such as Ethereum, every smart contract is a Lego of sorts, one should be able to use the functionality available in one smart contract as a building block in the design of a new smart contract. Similar to other areas such as cryptography or software development, \emph{composability} of smart contracts is crucial to allow modular designs that build upon other, well-studied smart contracts providing formally proven security and privacy guarantees. Assume that a user $u_i$ has committed data $d_i$ to our data tumbling layer and later wants to use the (encrypted) version of the data in different applications. Such data should be poised for use as input for another smart contract without the need for further processing, while concurrently, the data tumbling layer should be capable of verifying data ownership and preventing \emph{unauthorized} use of the data.

Hence, in the data tumbling problem, the following properties are of interest:
\begin{itemize}[leftmargin=0.5cm]
\item \emph{Correctness.} A user who commits to the data previously should be able to redeem the same data to use in other applications. 
\item \emph{No one-more redemption.} A user who commits $n$ data inputs, should not be able to redeem more than $n$ data outputs. 
\end{itemize}


\subsection{Our Contributions}
In summary, our main contributions are as follows:
\begin{itemize}[leftmargin=0.5cm]
    \item We introduce a cryptographic primitive termed the "Data Tumbling Layer" (DTL). We then generalize this to accommodate arbitrary data. We formalize crucial properties of DTL such as \emph{theft prevention}, \emph{non-slanderability}, \emph{unlinkability}, and \emph{no one-more redeeming}.  
    \item We provide two concrete DTL constructions: one for a fixed data input and another for arbitrary data input. The security and privacy of these constructions are cryptographically proven. 
    \item We demonstrate that both DTL constructions satisfy the composability property by utilizing DTL as a building block to construct applications with unlinkability features such as unlinkable fixed-amount payment,  unlinkable confidential payment, and unlinkable confidential weighted voting.
    The goal of composability in this work is to have a single contract that can be used to 
    provide unlinkability for different applications. 
    \item Finally, we implement the proposed applications, demonstrating that they are practical to run on existing blockchains like Ethereum.
\end{itemize}

%% file: relwork.tex
\section{Related Work}
\label{sec:related-work}

\begin{table*}[t]
\centering
\caption{Comparisons with previous works.}
\resizebox{2\columnwidth}{!}{%
    \begin{tabular}{c@{}cccccc@{}}
    \toprule
    & & \textbf{Anonymity Set} & \textbf{Confidential Amount} & \makecell{\textbf{Composability}\\ (Plug-and-Play \\with minimal changes)}  & \textbf{Availability} &\textbf{Applications} \\
    \midrule
    \multirow{4}{*}{\makecell{Privacy-preserving \\ Cryptocurrencies}} & 
    Monero & $ <2^4$ & \fullcirc & \emptycirc & \fullcirc  & Unlinkable and Confidential Payment \\
    & Zerocoin & Any & \emptycirc & \emptycirc &\fullcirc & Unlinkable and Confidential Payment \\
    & Zerocash & Any & \fullcirc & \emptycirc &\fullcirc & Unlinkable and Confidential Payment \\
    & Quisquis, Monero & $2^4$ & \fullcirc & \emptycirc & \fullcirc  & Unlinkable and Confidential Payment \\
    & Veksel & Any & \fullcirc & \emptycirc & \fullcirc  & Unlinkable and Confidential Payment \\
    \hline
    \multirow{4}{*}{\makecell{Off-chain Privacy\\ Solutions}} 
    & TumbleBit & Any & \emptycirc  & \emptycirc & \emptycirc  & Unlinkable Fixed-Amount Payment \\
    & A2L, A2L+ & Any & \emptycirc  & \emptycirc & \emptycirc  & Unlinkable Fixed-amount Payment \\ 
    & Blindhub  & Any & \fullcirc & \emptycirc & \emptycirc  & Unlinkable and Confidential Payment \\
    & Accio & Any & \fullcirc & \emptycirc & \emptycirc & Unlinkable and Confidential Payment \\\hline
    \multirow{5}{*}{\makecell{On-chain Privacy \\ Solutions}} 
    & CoinShuffle/CoinShuffle++ & Small & \emptycirc & \emptycirc  & \emptycirc  & Unlinkable Fixed-Amount Payment \\
    & Mobius & $<10$ & \emptycirc & \halfcirc  & \fullcirc & Unlinkable Fixed-Amount Payment \\
    & AMR & Any & \emptycirc & \halfcirc & \fullcirc  & Unlinkable Fixed-Amount Payment \\
    & Anonymous Zether & $<2^5$ & \fullcirc & \halfcirc & \fullcirc & Unlinkable and Confidential Payment \\
    \hline
    & \multirow{3}{*}{Ours} & \multirow{3}{*}{Any} & \multirow{3}{*}{\fullcirc} & \multirow{3}{*}{\fullcirc}  & \multirow{3}{*}{\fullcirc} & \multirow{3}{*}{\makecell{Unlinkable Fixed-Amount Payment (\added{\ref{subsec:fixed-amount}}),\\
    Unlinkable and Confidential Payment~(\added{\ref{sub:ddmzether}}),\\
    Unlinkable Weighted Voting~(\added{\ref{sub:ddmvoting}})}} \\ 
    \\ 
    \\
    \bottomrule
\end{tabular}
}
                                                            \end{table*}
This work aims to develop a data tumbling layer for blockchains, offering
security, privacy, and composability. This layer would allow for seamless
integration with blockchain applications, ensuring the unlinkability of
transactions. Achieving these properties simultaneously presents a significant challenge.
Below, we explore several existing attempts to design solutions that fulfill
these requirements.

\pparagraph{Private By Design Blockchain} 
Cryptocurrencies,
such as Monero~\cite{noether2014review}, Zerocash~\cite{sasson2014zerocash}, 
Zerocoin~\cite{miers2013zerocoin} and Quisquis~\cite{asiacrypt-2019-quisquis}, 
have been designed with tailored functionalities to provide privacy guarantees. 
Despite the robust privacy protections these cryptocurrencies provide,
private-by-design blockchains are typically restricted to transactions that
only involve asset transfers. Consequently, in such environments, the concept
of composability holds limited significance, as there are no other native
applications to interact with.
Also, each of these blockchains only supports a specific suite of cryptographic
primitives (e.g., Zcash is restricted to transaction authorization based on zero-knowledge proofs), making it even more difficult to interoperate with other applications.  
Like Zcash, Veksel~\cite{veksel2022} is a blockchain design that utilizes an optimized cryptographic suite, achieving larger anonymity sets and reduced transaction sizes. However, if deployed on Ethereum, its practical performance would be costly due to its reliance on RSA accumulators and slower zero-knowledge proof verification compared to Zcash. 

\pparagraph{Add-on Privacy Solutions: Off-chain Solutions} 
Another approach is to rely on an off-chain centralized server for tumbling, including
Tumblebit~\cite{heilman2017tumblebit}, A2L~\cite{tairi2021a2l}, A2L+~\cite{glaeser2022a2l2}, BlindHub~\cite{qin2023blindhub} and Accio~\cite{Ge2023accio}, which mainly utilize a centralized off-chain server for
mixing user funds. These designs offer a lower level of trust in the off-chain
server compared to alternatives like Mixcoin and Blindcoin, as they prevent the
server from misappropriating participant funds. Nevertheless, centralized
tumbling methods fall short in guaranteeing availability, as the central system
can block client deposits.

\pparagraph{Add-on Privacy Solutions: On-chain Solutions}
On the one hand, decentralized
tumbler approaches like Coinshuffle and Coinjoin have been designed as a peer-to-peer protocol overlay over existing cryptocurrencies. They aim to solve the availability issue with off-chain solutions by
enabling users to interact and create transactions that conceal the identities
of senders and recipients. However, ensuring the presence of participants and
their interaction poses challenges and could potentially expose privacy
vulnerabilities.
Additionally, the composability of these solutions presents obstacles, as funds
are typically restricted to the tumbling operation, limiting users'
ability to freely use those funds to interact with other on-chain applications.

On the other hand, there is also a line of work on add-on solutions for on-chain privacy developed on blockchains
that support smart contracts. They leverage the expressiveness of smart contracts to
implement a wide array of cryptographic primitives. Notably, Meiklejohn
\textit{et al.}~\cite{meiklejohn2018mobius} introduced an Ethereum-based mixing
service named
M\"{o}bius, utilizing linkable ring signatures and stealth address mechanisms,
similar to those in Monero, to conceal the identities of the sender and
recipient. However, the anonymity M\"{o}bius provides is capped by the ring's size,
with the withdrawal transaction's gas cost rising linearly with the ring size.

While more recent solutions have emerged, they face various limitations. 
Zapper~\cite{zapper2022}, which provides both data and identity privacy, requires a custom assembly language (Zasm) and cannot be directly deployed on Ethereum due to EVM incompatibility. Popular on-chain mixers like Railgun~\cite{railgun-protocol} 
can be viewed as specific instances of fixed-data tumbling, but they lack formal security analysis.

B\"{u}nz \textit{et al.}~\cite{zether-bunz-2020} developed Zether, a protocol for private payments on 
Ethereum, which conceals client balances using ElGamal encryption. The downside
is that Zether's transaction fees are high (7.8 million gas), and it still reveals
the sender and receiver of the transaction. To remedy this, Anonymous Zether~\cite{diamond2020many} 
was proposed by Diamond, enhancing privacy but still incurs high costs. Le \emph{et al.}~\cite{le2020amr} proposed AMR, an unlinkable fixed-amount payment. 
It can be integrated with existing DeFi protocols such as Aave~\cite{aave} or Compound~\cite{compound}, aiming to motivate user participation to expand the anonymity set. 
However, AMR only serves as an ad-hoc solution tailored for specific applications.

In contrast, our work provides a comprehensive approach: we formally define DTL with cryptographically proven security properties, extend it to handle arbitrary data types, and demonstrate how it can be modularly composed with different applications. This systematic treatment enables privacy-preserving smart contract development with formal security guarantees.

The key distinction is that while existing works offer tradeoffs between anonymity set size, availability, and application specificity, none provide generic \emph{composability}. Our work is the first to deliver this feature, demonstrating its utility across multiple applications.
\begin{inparaenum}
    \item unlinkable fixed-amount payments; 
    \item Unlinkable and confidential payments; 
    and, different from previous works, \item unlinkable weighted voting. 
\end{inparaenum}



%% file: preliminaries.tex
\section{Preliminaries}
\label{sec:prelim}


\paragraph{Notation}
We denote the security parameter by $\sparam$ and a negligible function in $\sparam$ by 
$\mathtt{negl}(\lambda)$.
We let $[n]$ denote the set $\{1,\dots, n\}$.
We use $(x_i)_{i=1}^{n}$ to denote a list of $n$ elements, $(x_1, \dots, x_n)$.
We let PPT denote probabilistic polynomial time.
We define the security properties of our protocol as games written in pseudocode.
These games invoke a probabilistic-polynomial time (PPT) adversary, \adv, with
access to some oracles.

{

\pparagraph{\zksnark}
A zero-knowledge Succinct Non-interactive ARgument of Knowledge
(\zksnark) is a ``succinct'' non-interactive zero-knowledge proofs (NIZK) for arithmetic circuit satisfiability. 
For a field $\FF$, an arithmetic circuit $C$ takes
as inputs elements in $\FF$ and outputs elements in $\FF$. 
We adopt a similar definition from Zerocash~\cite{sasson2014zerocash}
to define the arithmetic circuit satisfiability problem.
An arithmetic circuit satisfiability problem of a circuit ${C}:\FF^n\times\FF^h\rightarrow \FF^l$
is captured by the relation $R_C=\set{(\statement,\witness)\in \FF^n\times \FF^h: C(x,\witness) = 0^l}$, with 
the language $\mathcal L_C = \set{\statement \in \FF^n~|~\exists~\witness \in \FF^h~s.t~C(\statement,\witness) = 0^l}$. 



\begin{definition}[\zksnark~\cite{groth2016size}]
\label{def:zksnarks}
\zksnark, $\Pi$, for arithmetic circuit satisfiability is a triple of 
algorithms $(\mathtt{Setup}, \mathtt{Prove}, \mathtt{Verify})$: 
\begin{itemize}[leftmargin=0.5cm]
    \item $(\pk, \vk) \leftarrow \fun{Setup}(\sparam, C)$ 
    takes as input the security parameter and the arithmetic circuit $C$, outputs
    a proving key $\pk$, and a verification key $\vk$. The public parameters, $\pp$, is given implicitly
    to both proving and verifying algorithms.
    \item $\pi\leftarrow \fun{Prove}(\pk,\statement,\witness)$ takes as input the
    evaluation key $\ek$ and $(x, \witness) \in R_C$, outputs a proof 
    $\pi$ for the statement $x \in \mathcal{L}_C$
    \item $0/1\leftarrow\fun{Verify}(\vk, \statement, \pi)$ takes as input the verification
    key $\vk$, the public input $x$, the proof $\pi$,  outputs $1$ if $\pi$ is valid proof
    for  $x \in \mathcal{L}_C$. 
\end{itemize}
\end{definition}

Apart from \emph{correctness, soundness,} and \emph{zero-knowledge} properties, 
a \zksnark requires two additional properties, \textit{succinctness} and 
{\textit{simulation extractability}}.


\pparagraph{Commitment Scheme} 
A commitment scheme contains two operations: committing and revealing. During the committing operation, a user commits to selected values while concealing them from others. The user can choose to reveal the committed value during the revealing operation.  
\begin{definition}[Commitment Scheme]
\label{def:commitscheme}
  A commitment scheme, $\commit$ defined over $(\mathcal{X}, \mathcal{Y})$,
  consists of two algorithms: 
 \begin{itemize}[leftmargin=0.5cm]
      \item $\pp_{\mathsf{commit}} \leftarrow \fun{Setup}(\sparam)$ takes as
          input the security parameter, $\sparam$, it outputs the public
          parameters are implicit input to all subsequent algorithms.
      \item $\cm \leftarrow \fun{Commit}(m; r)$ accepts a message $m$ and a
          secret randomness $r$ as inputs and returns the commitment string
          $\cm$. 
      \item $0/1 \leftarrow \fun{Verify}(m, \cm; r)$ accepts a message $m$, a
          commitment $\cm$ and a value $r$ as inputs, and returns $1$ if the
          commitment is opened correctly and $0$ otherwise.
  \end{itemize}
\end{definition}
\noindent A commitment scheme should satisfy two security requirements:
(i) \emph{Binding:} Except for a negligible probability, no
probabilistic-polynomial time adversary can
efficiently produce $\cm$, $(m_1, r_1)$ and $(m_2, r_2)$ such that
$\fun{Verify}(m_1, \cm; r_1) = \fun{Verify}(m_2, \cm; r_2) = 1$ and $m_1 \neq
m_2$,
(ii) \emph{Hiding:} Except for a negligible probability, $\cm$ does not reveal
anything about the committed data.
{We can instantiate $\com$ with a hash function $h$ in the random oracle model.}


\pparagraph{Merkle Tree}~\label{pp:merkle-tree}
In this work, we focus solely on the Merkle tree, which is an authenticated data structure used for set membership.
A Merkle tree leverages a collision-resistant hash function to construct
the data structure. A Merkle tree consists of four algorithms that work as
follows:

\begin{definition}[Merkle Tree]
\label{def:ads}
A Merkle tree, $\mathsf{MT}$ defined over $(\mathcal{X}^n, \mathcal{Y})$, consists of the following algorithms: 
\begin{itemize}[leftmargin=0.5cm]
     \item $\rt\leftarrow\initialize(\sparam, \seqx)$ takes the security parameter and
         a list $\seqx$ as inputs, where $x_i \in \mathcal{X}$, $n$ is a power of $2$, and
         the algorithm outputs $\rt \in \mathcal{Y}$.

     \item $\adsproof \leftarrow \prove(j, x, \seqx)$ takes an element $x$, an
         index,  $j\in [n]$, and a sequence $\seqx$ as inputs, and outputs the
         proof $\pi$, which can prove that $x$ is the $j^{th}$ element in the
         sequence $\seqx$. 
    
    \item $0/1\leftarrow \mkverify(j, x, \rt, \adsproof)$ takes an index $j \in [n]$,
        an element $x$, a root, $\rt\in \bset^\lambda$ and a proof $\adsproof$ as
        inputs. The algorithm outputs $1$ if $x$ is the $j^{th}$ element in
        \seqx and $0$ otherwise.

  \end{itemize}

\end{definition}


   A Merkle tree should satisfy \emph{correctness} and \emph{security}. 
   For the formal
   definitions of these properties, we refer to the cryptography introduction
   book of Boneh and Shoup~\cite{boneh2020graduate}. 
   
\pparagraph{Tagging Scheme} We use the definition of the tagging
scheme from~\cite{lai2019omniring}. 
However, our tagging scheme does not
require the homomorphic property as defined in their protocol. Hence, we
modified the security experiments to not use the related input oracles. 
\begin{definition}[Tagging Scheme~\cite{lai2019omniring}]
\label{def:tagging-scheme}
A tagging scheme, $\mathsf{TAG}$ defined over $(\mathcal{K}, \mathcal{P}, \mathcal{T})$, is a
triple of algorithms, where: 
\begin{itemize}[leftmargin=0.5cm]
    \item $\tagpp\gets \tagsetup(\sparam)$: On input the
        security parameter $\sparam$, it outputs the public parameter,
        $\tagpp$. The public parameters are implicitly input to all
        subsequent algorithms. 
    \item $\dcmcoinpk \leftarrow \tagkeygen(\dcmcoinsk)$: On input the key
        $\dcmcoinsk$ from the secret key space~$\mathcal{K}$, it \emph{deterministically} outputs a public key
        $\dcmcoinpk \in \mathcal{P}$.
    \item $\key{tag} \leftarrow \tageval(\dcmcoinsk)$: On input the secret
        key $\dcmcoinsk$, it \emph{deterministically} outputs the tag  $\key{tag} \in \mathcal{T}$. 
\end{itemize}
\end{definition}
The tagging scheme should satisfy \emph{one-wayness}, \emph{collision-resistance}, and \emph{pseudo-randomness}. 

\pparagraph{Encryption Scheme} For arbitrary data construction, we need to use encryption scheme defined as follows. 
\begin{definition}
A public-key encryption scheme,
$\mathtt{E}$, defined over $(\mathcal{M}, \mathcal{C})$, works as follows: 
\begin{itemize}[leftmargin=0.5cm]
    \item  $(\dk, \ek) \leftarrow \fun{KGen}(\sparam)$: On input the security
        parameter $1^\lambda$, the algorithm outputs a pair of decryption
        and encryption keys, $(\dk, \ek)$. 
    \item  $c \leftarrow \fun{Enc}(\ek, m; r)$: On input an encryption key, \ek,
    a message, $m \in \mathcal{M}$, and randomness $r$, the algorithm outputs a ciphertext, $c \in \mathcal{C}$. 
    \item  $m \leftarrow \fun{Dec}(\dk, c)$: On input a decryption key, $\dk$
    and a ciphertext $c \in \mathcal{C}$, the algorithm outputs a message, $m\in \mathcal{M}$.
\end{itemize}
\end{definition}
The encryption scheme should be correct and satisfy ciphertext
indistinguishability under chosen-plaintext attacks (IND-CPA). For certain applications, 
we require the encryption scheme to be \emph{additively homomorphic}. 
We assume there is a deterministic function $\fun{derive}$, which takes as input
the decryption key and outputs the encryption key, i.e. $\ek = \fun{derive}(\dk)$.

%% file: problem-statement.tex
\section{Data Tumbling Layer}
\label{sec:problem-statement}

In this section, we first define the notion of the Data Tumbling Layer. Then we introduce the security and privacy notions of interest. 

\begin{figure*}[!h]
    \centering
    \begin{pcvstack}[boxed]
    \begin{pchstack}
    \begin{pcvstack}
    \procedure[space=keep, bodylinesep=3pt, codesize=\footnotesize]{$\dcmcreatecoin\oracle({\dcmdata}) $}{
        (\dcmcoinpk, \dcmcoinsk) \gets \dcmcreatecoin({\dcmdata})\\
        \dcmcoinlist \define \dcmcoinlist \cup \set{(\dcmcoinpk, \dcmcoinsk)}\\
        \pcreturn \dcmcoinpk
    }
    \pcvspace
    \procedure[space=keep, bodylinesep=3pt, codesize=\footnotesize]{$\dcmexpds_{\adv}(\sparam) $}{
    \pp \gets \dcmsetup(\sparam),
    \dcmcoinlist \define \emptyset, \dcmredeemlist \define \emptyset\\
    \{(\nulli_j, \dcmproof_j, m_j)\}_{j\in [n+1]}, \dcmncoins \gets \adv(\pp)\\
    \dcmcoinaccstate \leftarrow \dcmaccumulatecoin(\dcmncoins)\\
    \bit_0 \define \wedge_{j\in [n+1]}\dcmverify(\dcmcoinaccstate, \nulli_j, \dcmproof_j, m_j)\\
    \bit_1\define \forall i \neq j \in [n+1], \nulli_i \neq \nulli_j\\
    \pcreturn \bit_0 \land \bit_1 
    }
    \pcvspace
    \procedure[space=keep, bodylinesep=3pt, codesize=\footnotesize]{$\dcmexpunlink_{\adv}(\sparam) $}{
        \pp \gets \dcmsetup(\sparam),
        \dcmcoinlist \define \emptyset, \dcmredeemlist \define \emptyset\\
        m \sample \mathcal{M}\\
        \dcmdata_0, \dcmdata_1 \leftarrow \adv(\pp, m)\\
        ({\dcmcoinpk_0}, {\dcmcoinsk_0}) \gets \dcmcreatecoin(\dcmdata_0)\\
        ({\dcmcoinpk_1}, {\dcmcoinsk_1}) \gets \dcmcreatecoin(\dcmdata_1) \\
        (\nulli_0, \dcmproof_0) \gets \dcmredeemcoin((\dcmcoinpk_0, \dcmcoinpk_1), {\dcmcoinsk_0}, m)\\
        (\nulli_1, \dcmproof_1) \gets \dcmredeemcoin((\dcmcoinpk_0, \dcmcoinpk_1), {\dcmcoinsk_1}, m)\\
        \bit \sample \set{0,1}\\
        \bit' \gets \adv^{\dcmcreatecoin\oracle, \dcmredeemcoin\oracle}(\dcmcoinpk_0, \dcmcoinpk_1, (\nulli_{0 \oplus \bit}, \dcmproof_{0 \oplus \bit}), (\nulli_{1\oplus \bit}, \dcmproof_{1 \oplus \bit}))\\
        \bit_0 \define \bit = \bit' \\
        \bit_1 \define (\cdot, {\dcmcoinpk_0}, \cdot, \cdot, \cdot) \notin \dcmredeemlist \land (\cdot, {\dcmcoinpk_1},\cdot, \cdot, \cdot) \notin \dcmredeemlist\\
        \pcreturn \bit_0 \land \bit_1
        }
    \end{pcvstack}
    \pchspace
    \begin{pcvstack}
    \procedure[space=keep, bodylinesep=3pt, codesize=\footnotesize]{$\dcmredeemcoin\oracle(\dcmncoins, i, m) $}{
        \pcif \nexists (\dcmcoinpk^*, \dcmcoinsk^*) \in \dcmcoinlist | \dcmcoinpk^* = \dcmcoinpk_i~\pcreturn \bot \\
        \textbf{let}~ (\dcmcoinpk^*, \dcmcoinsk^*) \in \dcmcoinlist | \dcmcoinpk^* = \dcmcoinpk_i \\
        \dcmcoinaccstate 
        \gets \dcmaccumulatecoin(\dcmncoins)\\
        (\nulli_i, \dcmproof_i) \gets \dcmredeemcoin(\dcmncoins, \dcmcoinsk^*, m)\\
        {\dcmredeemlist \define \dcmredeemlist \cup \set{(\dcmcoinaccstate, \dcmcoinpk^*, \nulli_i, \dcmproof_i, m)}}\\
        \pcreturn (\nulli_i, \dcmproof_i) 
    }
    \pcvspace
    \procedure[space=keep, bodylinesep=3pt,codesize=\footnotesize]{$\dcmexptheft_{\adv}(\sparam) $}{
    \pp \gets \dcmsetup(\sparam),
    \dcmcoinlist \define \emptyset, \dcmredeemlist \define \emptyset\\
    (\dcmncoins, \nulli, \dcmproof, m) \gets \adv^{\dcmcreatecoin\oracle, \dcmredeemcoin\oracle}(\pp)\\
    \dcmcoinaccstate \leftarrow \dcmaccumulatecoin(\dcmncoins)\\
    \bit_0 \define \dcmverify(\dcmcoinaccstate, \nulli, \dcmproof, m)\\
    \bit_1 \define  \forall \dcmcoinpk \in \dcmncoins, (\dcmcoinpk, \cdot) \in C\\
    \bit_2 \define (\dcmcoinaccstate, \cdot, \nulli, \cdot, m) \notin \dcmredeemlist\\ 
    \pcreturn \bit_0 \land \bit_1  \land \bit_2
    }
    \pcvspace
    \procedure[space=keep, bodylinesep=3pt, codesize=\footnotesize]{$\dcmnonslander_{\adv}(\sparam) $}{
        \pp \gets \dcmsetup(\sparam),
        \dcmcoinlist \define \emptyset, \dcmredeemlist \define \emptyset\\
        (\dcmcoinpk, \dcmcoinpk^*), (\nulli^*, \pi^*, m^*) \gets \adv^{\dcmcreatecoin\oracle, \dcmredeemcoin\oracle}(\pp)\\
        \dcmcoinaccstate \leftarrow \dcmaccumulatecoin((\dcmcoinpk, \dcmcoinpk^*))\\
        \bit_0 \define (\dcmcoinpk, \cdot)\in C\\ 
        m \sample \mathcal{M}\\
        (\nulli, \pi)\leftarrow\dcmredeemcoin\oracle((\dcmcoinpk, \dcmcoinpk^*), 1, m)\\
        \bit_1 \define \nulli = \nulli^*\\
        \bit_2 \define \dcmverify(\dcmcoinaccstate, \nulli^*, \dcmproof^*, m^*) \\
        \bit_3 \define (\dcmcoinaccstate, \cdot, \nulli, \cdot, m^*) \notin \dcmredeemlist\\ 
        \pcreturn \bit_0 \land \bit_1 \land \bit_2 \land \bit_3
    }
    \end{pcvstack}
    \end{pchstack}
    \end{pcvstack}
    \caption{Definition of oracles and various experiments.}
    \label{fig:exps}
\end{figure*}

\begin{definition}[Data Tumbling Layer (DTL)]
\label{def:dcm}
A data tumbling layer (\DDM) is a tuple of algorithms (\dcmsetup,
\dcmcreatecoin, \dcmaccumulatecoin, \dcmredeemcoin, \dcmverify) as
follows:

\begin{itemize}[leftmargin=0.5cm]
    \item $ \pp \gets \dcmsetup(\sparam) $: On input the security parameter
       $\sparam$, it outputs the public parameters $\pp$. The
       public parameters $\pp$ are implicitly input to all subsequent algorithms.
    
    \item $ (\dcmcoinpk, \dcmcoinsk) \gets \dcmcreatecoin({\dcmdata})$: On
        input the public parameters $\pp$ and the  data, $\dcmdata \in \mathcal{D}$,
        outputs the public ($\dcmcoinpk$) and private ($\dcmcoinsk$) key
        representations of {$\dcmdata$}.  
    
    \item $ \dcmcoinaccstate \gets
       \dcmaccumulatecoin(\dcmncoins)$: On input a sequence of public keys
       $\dcmncoins$, it outputs a state $\dcmcoinaccstate$. 

    \item $\set{(\nulli,\dcmproof), \bot} \gets \dcmredeemcoin(\dcmncoins,
        \dcmcoinsk, m)$: On input a list of public keys $\dcmncoins$, {an
        application-dependent message $m \in \mathcal{M}$}, and the private key of a piece of data 
        $\dcmcoinsk$, it outputs a pair of tag, proof, $(\nulli, \dcmproof)$, or $\bot$.
    
    \item $\bit \gets \dcmverify(\dcmcoinaccstate, \nulli, \pi, m)$: On input
        a state $\dcmcoinaccstate$, a tag, $\nulli$, a proof, $\dcmproof$, and a message
        $m$, it outputs a bit $\bit$.

\end{itemize}
We say that a DTL is \emph{correct} if for all $\lambda \in \NN$, all $n\in
\poly(\lambda)$ and all $i$, s.t. $1 \leq i \leq n$, any {$\dcmdata_i \in \mathcal{D}, m \in
\mathcal{M}$} it holds that:
{\footnotesize 
\begin{equation*}
    \begin{aligned}
    \Pr \left[
         \begin{aligned}
          &\pp \gets \dcmsetup(\sparam), 
         \\& (\dcmcoinpk_i, \dcmcoinsk_i) \gets \dcmcreatecoin(\dcmdata_i) \text{ for } \forall i \in [n],
         \\& \dcmcoinaccstate \gets \dcmaccumulatecoin(\dcmncoins) , 
         \\& (\nulli_i, \dcmproof_i) \gets \dcmredeemcoin(\dcmncoins, \dcmcoinsk_{i}, m),
         \\& \text{ s.t. }
         \forall i \in [n], \dcmverify(\dcmcoinaccstate, \nulli_i, \dcmproof_i, m) = 1  
         \end{aligned}
        \right] = 1
    \end{aligned}
\end{equation*}
}
\end{definition}


Next, we discuss the notions of interest for a DTL scheme. 

\pparagraph{No One-More Redeeming} 
We require a notion called \emph{no one-more redeeming} for a DTL scheme. 
Intuitively, this property guarantees that no user (including the adversary)
can redeem a public key \dcmcoinpk more than once, even if he knows
the corresponding \dcmcoinsk. This property thereby captures the idea that an
adversary cannot use a DTL scheme to redeem more than the number of keys that it controls. 
We formally describe this property in~\cref{def:dcm-double-spending}. 

\begin{definition}[DTL Security: No One-more Redeeming]
\label{def:dcm-double-spending}
    We say that a DTL scheme is \emph{secure against one-more redeeming} if for all $\lambda \in \NN$ there exists a negligible function $\negl(\lambda)$ such that 
$
\Pr\left[ \dcmexpds_{\adv}(\sparam) \right]\leq \negl(\lambda)
$, 
where the experiment $\dcmexpds$ is defined in~\cref{fig:exps}. 
\end{definition}

\pparagraph{Theft Prevention} We require another security notion called \emph{theft prevention} for a DTL scheme. 
This property guarantees that it is infeasible for an adversary to successfully redeem a public key \dcmcoinpk for which they do not know the corresponding private key \dcmcoinsk. 
This property thereby captures the idea that an adversary should not be able to steal a public key from an honest user. 
We formally describe this property in~\cref{def:dcm-theft}. 

\begin{definition}[DTL Security: Theft prevention]
\label{def:dcm-theft}
    We say that a DTL scheme is \emph{secure against theft} if for all $\lambda \in \NN$ there exists a negligible function $\negl(\lambda)$ such that 
$
\Pr\left[ \dcmexptheft_{\adv}(\sparam) \right] \leq \negl(\lambda)
$, 
where $\dcmexptheft$ is defined in~\cref{fig:exps}. 

\end{definition}

\pparagraph{Non-slanderability} Another security notion required in a DTL is \emph{non-slanderability}. 
This property guarantees that an adversary cannot prevent an honest user from successfully redeeming their previously committed data. 
This property thereby captures the idea that the adversary cannot launch a denial of service attack against an honest user. 
Note that non-slanderability differs from theft prevention in that in the former, the adversary is required to successfully redeem the victim's public key $\dcmcoinpk$ whereas, in the latter, the adversary does not necessarily redeem the victim's \dcmcoinpk and yet can prevent the victim from redeeming their $\dcmcoinpk$.  
We formalize this notion in~\cref{def:dcm-nonslanderability}. 

\begin{definition}[DTL Security: Non-slanderability]
\label{def:dcm-nonslanderability}
    We say that a DTL scheme is \emph{non-slanderable} if for all $\lambda
    \in \NN$ there exists a negligible function $\negl(\lambda)$ such that 
$
\Pr\left[ \dcmnonslander_{\adv}(\sparam) \right] \leq \negl(\lambda)
$, 
where $\dcmnonslander$ is defined in~\cref{fig:exps}. 
\end{definition}

\pparagraph{Unlinkability} The privacy notion required by a DTL is \emph{unlinkability}. 
Intuitively, this property guarantees that an adversary cannot link a \dcmredeemcoin operation from an honest user to the public key \dcmcoinpk being redeemed. 
We formally capture this notion in~\cref{def:dcm-unlinkability}. 

\begin{definition}[DTL Privacy: Unlinkability]
\label{def:dcm-unlinkability}
We say that a DTL scheme provides \emph{unlinkability} if for all $\lambda \in
\NN$ there exists a negligible function $\negl(\lambda)$ such that 
$
    \left|\Pr\left[\dcmexpunlink_\adv(\sparam)\right]-1/2\right| \leq \negl(\lambda)
$, 
where $\dcmexpunlink$ is defined in~\cref{fig:exps}. 
\end{definition}

%% file: construction.tex
\begin{figure*}[t]
    \centering
    \begin{pchstack}[boxed]
        \begin{pcvstack}
            \procedure[space=keep, bodylinesep=3pt, codesize=\footnotesize]{$\dcmsetup(\sparam)$}{
                \text{Let: } \\
                \t \Pi \text{ be a \zksnark scheme,} \\
                \t \tree \text{ be a Merkle tree instance,} \\
                \t \tagscheme \text{ be a tagging scheme,} \\
                \t \dcmdata_{\mathsf{fixed}} \text{ be the default data}\\
                \tagpp \gets \tagscheme.\tagsetup(\sparam)\\
                (\pk, \vk) \gets \Pi.\setup(\sparam,  C_{\mathsf{fixed}})\\
                \pp \define (\pk, \vk, \tagpp, \dcmdata_{\mathsf{fixed}})\\
                \pcreturn \pp
            }
        \end{pcvstack}
        \pchspace
        \begin{pcvstack}
            \procedure[space=keep, bodylinesep=3pt, codesize=\footnotesize]{$\dcmaccumulatecoin(\{{\dcmcoinpk_i}\}_{i\in [n]})$}{
                \dcmcoinaccstate \gets \tree.\initialize(\sparam, {\dcmncoins})\\
                \pcreturn \dcmcoinaccstate
            }
            \pcvspace
            \procedure[space=keep, bodylinesep=3pt, codesize=\footnotesize]{$\dcmverify({\dcmcoinaccstate}, \nulli, \pi, m$)}{
                (\zkproof, \dcmdata) \leftarrow \pi \\
                \statement \define ({\dcmcoinaccstate}, \nulli, m)\\
                \pcreturn \zkverify(\vk, \statement, \zkproof)
            }
            \pcvspace
            \procedure[space=keep, bodylinesep=3pt, codesize=\footnotesize]{$\dcmcreatecoin()$}{
                k \sample \bset^\lambda, r \sample \bset^\lambda\\
                \dcmcoinsk \define (k, r)\\
                \dcmcoinpk \gets \tagscheme.\tagkeygen(\dcmcoinsk)\\
                \pcreturn (\dcmcoinpk, \dcmcoinsk)
            }
        \end{pcvstack}
        \pchspace
        \begin{pcvstack}
             \procedure[space=keep, bodylinesep=3pt, codesize=\footnotesize]{$\dcmredeemcoin(\{{\dcmcoinpk_i}\}_{i\in [n]}, {\dcmcoinsk}, m$)}{
                    \dcmcoinpk \leftarrow \tagscheme.\tagkeygen(\dcmcoinsk)\\
                    \pcif \dcmcoinpk \notin \{{\dcmcoinpk_i}\}_{i\in [n]}~\pcreturn \bot\\
                    \textbf{let}~ j \in 1,\ldots,n : {\dcmcoinpk_j} = \dcmcoinpk\\
                    \dcmcoinaccstate \gets \tree.\initialize(\sparam, {\dcmncoins})\\
                    \adsproof \gets \tree.\prove(j, \dcmcoinpk, \dcmncoins)\\
                    \nulli \gets \tagscheme.\tageval(\dcmcoinsk) \\
                    \statement \define (\dcmcoinaccstate, \nulli, m)\\
                    \witness \define (j, \dcmcoinsk, \adsproof)\\
                    \zkproof \gets \zkprove(\pk, \statement, \witness)\\
                    \pi \define \zkproof\\
                    \pcreturn (\nulli, \pi)
            }
        \end{pcvstack}
    \end{pchstack}
    \caption{Our construction of \fDDM~for fixed data.\label{fig:construction-fixed-amt}}
\end{figure*}

\section{DTL Constructions}
\label{sec:DDM-construction}
In this section, we provide two constructions for DTL. One is for a fixed data
input scenario where input data is identical, and the other is for
arbitrary data where the input data can be arbitrary.

\subsection{$\DDM_\fixed$: DTL Construction for Fixed Data}
\label{sec:dcm-construction-simple}

For certain applications, like fixed-amount mixers or equal-weight voting systems, all users will have identical data (for instance, a set amount of money or voting eligibility). 
The primary aim here is to disassociate the data from its original depositor. Realizing a DTL with identical data is fairly straightforward. 

The detailed construction of $\fDDM$ is in \cref{fig:construction-fixed-amt}. Below, we explain in detail each segment of this construction.


\pparagraph{Parameter Setup}
To construct $\DDM_\fixed$, one needs to instantiate a secure \zksnark scheme 
$\Pi$, a Merkle tree \tree, 
a secure tagging scheme \tagscheme, and a predefined data $\dcmdata_\fixed$. 
{We run $\tagscheme$ setup algorithm to generate the public parameter for the
tagging scheme, $\pp_{\tagscheme}$. 
{Additionally, a one-time setup is required
for the \emph{statement} asserting the ownership of a secret key 
that corresponds to one of the public keys that have been accumulated into
the Merkle tree.}

Specifically, we have: 
$\statement \define (\dcmcoinaccstate, \nulli, m)$ and witnesses $\witness
\define (i, \dcmcoinsk, \adsproof)$, with the relation defined as follows:
{\footnotesize 
\begin{equation}~\label{eq:fixed}
        R_{C_\mathsf{fixed}} \define \left\{
        \begin{aligned}
        &(\statement\define (\dcmcoinaccstate, \nulli, m); \witness\define (i, \dcmcoinsk, \adsproof)): 
        \\& \tree.\mkverify(i, \dcmcoinpk, \dcmcoinaccstate, \adsproof)  \land  
        \\& \dcmcoinpk = \tagscheme.\tagkeygen(\dcmcoinsk) \land
        \\& \nulli = \tagscheme.\tageval(\dcmcoinsk)
        \end{aligned}
                                   \right\}
\end{equation}
}
Finally, the algorithm outputs the public parameters, including the zero-knowledge proof key pair and the public parameter for the tagging scheme. 

\pparagraph{Creating Key Pair} This algorithm 
uniformly samples two random elements, $k, r$, from $\bset^{\lambda}$. 
It returns the secret key, $\dcmcoinsk = (k,r)$, and the public
key, $\dcmcoinpk \define \tagscheme.\tageval(\dcmcoinsk)$.

\pparagraph{Accumulating Keys} This algorithm 
takes a list of public keys as input and produces a succinct representation of them. The algorithm executes the $\tree.\initialize()$ function to accumulate those public keys
into a Merkle root, denoted as $\dcmcoinaccstate$ and subsequently outputs $\dcmcoinaccstate$.

\pparagraph{Redeeming Data} The primary objective of this algorithm is to generate a cryptographic proof that confirms a user's control over one of the keys from the provided list of public keys. This is coupled with linking a specific message, 
$m$, to the proof. 

More precisely,
the algorithm takes a list of public keys, 
$\dcmncoins$, a secret key, 
\dcmcoinsk, and an arbitrary message 
$m$ as its inputs.
If the public key corresponding to $\dcmcoinsk$ is not found within 
$\dcmncoins$, the algorithm returns $\bot$.
However, if a match is found, the algorithm triggers the 
\zksnark instance to generate the proof 
$\pi$. This proof verifies the statement detailed in~\cref{eq:fixed}.

Notably, a static message $m$ is also incorporated into the public statement. 
{
For example, in applications like unlinkable payment, this message may consist
of the address of the recipient and the fee paid for the tumbling contract.
}
This message cannot be altered by any adversary if the \zksnark instance has the simulation extractability property. 

\pparagraph{Verifying Redemption}
The algorithm runs
$\zkverify$ to verify the validity of the statement defined \cref{eq:fixed}.

\subsection{\aDDM: DTL Construction for Arbitrary Data}
\label{sub:ddm-arbitrary}

{
In this construction, it is crucial to encrypt the data when it is redeemed. If data inputs vary and are not encrypted during redemption, it could easily allow a public key to be linked to its proof, undermining the desired unlinkability. As such, we must leverage the commitment scheme's hiding property during key creation to obscure the data and implement an encryption scheme when the data is redeemed.
Furthermore, to facilitate the encryption process, the input message $m$ of the redemption function must include the public
encryption key, \ek, as its first element.
}

\input{protocols/dml-abitrary}

\pparagraph{Setting Up Parameters} 
Similar to \fDDM, one needs to instantiate a secure \zksnark scheme, $\Pi$, a Merkle tree, \tree, a secure tagging scheme, \tagscheme. 
Additionally, for $\aDDM$,  an IND-CPA secure encryption scheme, $\mathtt{E}$
and a \emph{hiding} and \emph{binding} commitment scheme $\commit$ are utilized.

Furthermore, a one-time setup is required for the statement.
This statement verifies that an individual has control over a secret key
corresponding to a public key integrated into the Merkle tree root. It also
ensures that the data used to compute the public key is identical to the data
being encrypted. Finally, we may need to prove that the data satisfies some predicate $P(\cdot)$  (e.g., a range proof).

To be more specific, we utilize \zksnark to address the subsequent hard relation:
for statement $\statement \define (\dcmcoinaccstate, \nulli, \ek, c)$ and
witnesses $\witness \define (\dcmcoinsk, \adsproof, r_{\mathsf{enc}})$ where
$\dcmcoinsk \define (k,r,\dcmdata)$.

{\footnotesize 
\begin{equation}~\label{eq:arbit}
        R_{C_\mathsf{arb}} \define \left\{
        \begin{aligned}
        &(\statement\define (\dcmcoinaccstate, \nulli, \ek, c) ; \witness \define (\dcmcoinsk, \adsproof, r_{\mathsf{enc}})): 
        \\& \tree.\mkverify(i, \dcmcoinpk, \dcmcoinaccstate, \adsproof)  \land
                                      \\& \nulli = \tagscheme.\tageval(\dcmcoinsk) \land 
                                      \\& \dcmcoinpk' = \tagscheme.\tagkeygen(\dcmcoinsk) \land 
                                      \\&\dcmcoinpk = \commit.\commitcreate(\dcmdata; \dcmcoinpk') \land 
                                      \\& c = \mathtt{E}.\mathsf{Enc}(\ek, \dcmdata; r_{\mathsf{enc}}) \land {P(\dcmdata) = 1}
        \end{aligned}
        \right\}
\end{equation}
}

Similar to \fDDM, we run $\tagscheme$ setup to generate the public parameter,
$\pp_{\tagscheme}$. Finally, the algorithm outputs as the public parameters
the zero-knowledge proof key pair (i.e., evaluation key and
verification key) and the public parameter for the tagging scheme. 

\pparagraph{Creating Key Pair} This algorithm takes \dcmdata as input and
outputs a pair of secret key and public key such that \dcmdata is binding to the public key. 
To accomplish this, the algorithm samples two elements $k,r$ from $\bset^{\lambda}$. It then executes $\tagscheme.\tageval$ using $(k,r)$ to obtain $\dcmcoinpk'$ and uses this key as randomness to compute the public key $\dcmcoinpk = \commit.\commitcreate(\dcmdata, \dcmcoinpk')$. Finally, the algorithm outputs, $\dcmcoinsk = (k,r, \dcmdata)$ and $\dcmcoinpk$.

\pparagraph{Accumulating Keys} This algorithm is identical to \fDDM.

\pparagraph{Redeeming Encrypted Data} 
Unlike $\fDDM$, the redeem proof now
includes the encryption of \dcmdata associated with the public key. This
algorithm demonstrates that the user has control over one of the keys, and that
the data linked to the public key is indeed encrypted within the ciphertext. To
achieve this, the message must contain the encryption key, $\ek$. During the
redeeming step, the algorithm encrypts the data associated with the encryption
key, then utilizes the \zksnark instance to verify that the ciphertext is an
encryption of the data linked to the public key.
Finally, the algorithm outputs the tag and proof pair consisting of the tag,
\nulli, and the proof $\pi$ which consists of the \zksnark proof, \zkproof, that verifies~\cref{eq:arbit} and the
ciphertext $c = \mathtt{E.Enc}(\ek, \dcmdata)$.

\pparagraph{Verifying Redemption} The verification takes as input the 
\nulli and the proof \dcmproof, consisting of the \zksnark proof \zkproof and the ciphertext $c$. 
The algorithm runs  $\zkprove$ to verify the validity of the statement defined
in~\cref{eq:arbit}.

The detailed construction of $\aDDM$ is defined in
\cref{fig:construction-arbitrary-amt}.

\subsection{Security Analysis}
\label{sec:security-analysis}
We formally prove the security of both $\fDDM$ and $\aDDM$. 
Due to the space limit, we shift the proofs to~\cref{sec:missing-proofs}.

\begin{reptheorem}{fixed}
Given the security of \zksnark instance, $\Pi$,
the Merkle tree, $\tree$, and the tagging scheme, $\tagscheme$
as defined in~\cref{sec:prelim}, then $\fDDM$~for fixed
input data satisfies \emph{correctness}, \emph{no one-more redemption},
\emph{theft prevention}, \emph{non-slanderability}, and \emph{unlinkability}. 
\end{reptheorem}

\begin{reptheorem}{arbit}
Given the security of \zksnark instance, $\Pi$, the Merkle tree, $\tree$,
the tagging scheme, $\tagscheme$
the commitment scheme, $\mathtt{COM}$
and the IND-CPA encryption scheme, $\mathsf{E}$ as defined
in~\cref{sec:prelim}, then $\aDDM$~for arbitrary input data satisfies
\emph{correctness}, \emph{no one-more redemption},
\emph{theft prevention}, \emph{non-slanderability}, and \emph{unlinkability}. 
\end{reptheorem}

%% file: protocols/dml-abitrary.tex
\begin{figure*}[h]
    \centering
    \begin{pchstack}[boxed]
        \begin{pcvstack}
            \procedure[space=keep, bodylinesep=3pt, codesize=\footnotesize]{$\dcmsetup(\sparam$)}{
                \text{Let: } \\
                \t \Pi \text{ be a zksnark scheme,} \\
                \t \tree \text{ be a Merkle tree instance,} \\
                \t \tagscheme \text{ be a tagging scheme} \\
                \t \mathtt{E} \text{ be an  encryption scheme }\\
                \t \commit \text{ be a commitment scheme}\\
                \tagpp \gets \tagscheme.\tagsetup(\sparam)\\
                (\pk, \vk) \gets \Pi.\setup(\sparam, C_{\mathsf{arb}})\\
                \pp \define (\pk, \vk, \tagpp)\\
                \pcreturn \pp
            }

        \end{pcvstack}

        \pchspace

        \begin{pcvstack}
            \procedure[space=keep, bodylinesep=3pt, codesize=\footnotesize]{$\dcmaccumulatecoin({\dcmncoins}$)}{
                \dcmcoinaccstate \gets \tree.\initialize(\sparam, {\dcmncoins})\\
                \pcreturn \dcmcoinaccstate
            }
            
            \pcvspace
            \procedure[space=keep, bodylinesep=3pt, codesize=\footnotesize]{\dcmverify(${\dcmcoinaccstate}, \nulli, \dcmproof, m$)}{
                \ek\leftarrow m\\
                (\zkproof, c) \leftarrow \dcmproof \\
                \statement \define ({\dcmcoinaccstate}, \nulli, \ek, c)\\
                \pcreturn \zkverify(\vk, \statement, \zkproof)
            }
            \pcvspace
            \procedure[space=keep, bodylinesep=3pt, codesize=\footnotesize]{$\dcmcreatecoin(\dcmdata)$}{
                k \sample \bset^\lambda, r \sample \bset^\lambda\\
                \dcmcoinsk' \define (k, r)\\
                \dcmcoinpk' \gets \tagscheme.\tagkeygen(\dcmcoinsk')\\
                \dcmcoinsk = (k,r,\dcmdata)\\
                \dcmcoinpk = \commit.\commitcreate(\dcmdata, \dcmcoinpk')\\
                \pcreturn (\dcmcoinpk, \dcmcoinsk)
            }
        \end{pcvstack}
    \pchspace
    \begin{pcvstack}
            \procedure[space=keep, bodylinesep=3pt, codesize=\footnotesize]{$\dcmredeemcoin({\dcmncoins}, {\dcmcoinsk}, m$)}{
                (k, r, \dcmdata) \gets {\dcmcoinsk}\\
                \dcmcoinsk' \define (k, r)\\
                \dcmcoinpk' \gets \tagscheme.\tagkeygen({\dcmcoinsk'})\\
                \dcmcoinpk \gets \commit.\commitcreate(\dcmdata; \dcmcoinpk')\\
                \pcif \dcmcoinpk \notin {\dcmncoins}~\pcreturn \bot\\
                \textbf{let}~ j \in \{1,\ldots,n\}: {\dcmcoinpk_j} = \dcmcoinpk\\
                \dcmcoinaccstate \gets \tree.\initialize(\sparam, {\dcmncoins})\\
                \adsproof \gets \tree.\prove(j, \dcmcoinpk, {\dcmncoins})\\
                \nulli \gets \tagscheme.\tageval(\dcmcoinsk')\\
                m = (\ek, \dots)  \\ 
                r_{\mathsf{enc}} \sample \bset^\lambda\\
                c \define \mathtt{E}.\fun{Enc}(\ek, \dcmdata; r_{\mathsf{enc}})\\
                \statement \define (\dcmcoinaccstate, \nulli, m, c)\\
                \witness \define (k, r, \dcmdata, \adsproof, r_{\mathsf{enc}})\\
                \zkproof \gets \zkprove(\pk, \statement, \witness)\\
                \dcmproof \define (\zkproof, c)\\
                \pcreturn (\nulli, \dcmproof)
            }
    \end{pcvstack}
    \end{pchstack}
\caption{Our construction of \aDDM~for arbitrary data.}
\label{fig:construction-arbitrary-amt}
\end{figure*}

%% file: overview.tex
\section{Composability of DTL with Blockchain Applications} 
\label{sec:solution-overview}

In this section, we showcase the composability of DTL, specifically, its
ability to provide unlinkability to numerous blockchain applications that
inherently lack this assurance.



\subsection{Our Generic Approach for Composability}

\pparagraph{System Model and Assumptions} 
The use of our data tumbler layer is orchestrated by a smart contract, $\cdtl$. 
Multiple users, who have access to the
tumbler contract, use it to tumble their data in a decentralized manner. More
concretely, users commit their data to the tumbler contract so that later only
the intended recipient (which can be the same as the depositing user) can
withdraw. An illustrative example  is in~\cref{fig:approach overview}.

We assume that users have enough funds to pay for the fees charged to
interact with the tumbler contract. Moreover, we assume authenticated,
confidential communication channels between any two users that interact with
each other outside the tumbler contract.  
The code of the smart contracts is trusted for
integrity, but not for privacy. This reflects the fact that, in practice, the
distributed ledger operators (e.g., miners) 
reach consensus on the correct execution of a call
to tumbler contract, but they have access to data used for such computation.
Users can be arbitrarily malicious and interact with the tumbler contract in any
fashion. They are only restricted to respect the API of the tumbler contract,
since otherwise, the call is straightforwardly rejected by the distributed
ledger operators.

\pparagraph{Additional Data Structures used in $C_{\mathsf{DTL}}$} To manage the
\emph{concurrency} of multiple commits, the smart contract must maintain
a deposit history, denoted as $\dpl$, which includes all commits~\footnote{
For the concrete implementation,
each deposit can be realized as a triggering event, allowing users to avoid
storage costs. Explicitly referencing $\dpl$ simplifies the explanation of our
protocol.
}. 
Additionally, it requires an accumulator history, $\drootlist$, consisting of
$k$ recent
Merkle roots. The reason is
that each data commit will alter the root of the Merkle tree; hence, storing
$k$ roots, allows $C_{\mathsf{DTL}}$ to handle up to $k$ concurrent data commits.
Depending on the speed of the blockchain, this parameter can be adjusted accordingly. 
We refer readers to \cite{le2020amr} for a detailed explanation. 
Finally, to prevent \emph{double redeeming}, $C_{\mathsf{DTL}}$ needs to maintain
a tag list, $\wnl$, storing all tags, $\nulli$, resulting from data redeemptions.

\begin{figure*}[t]
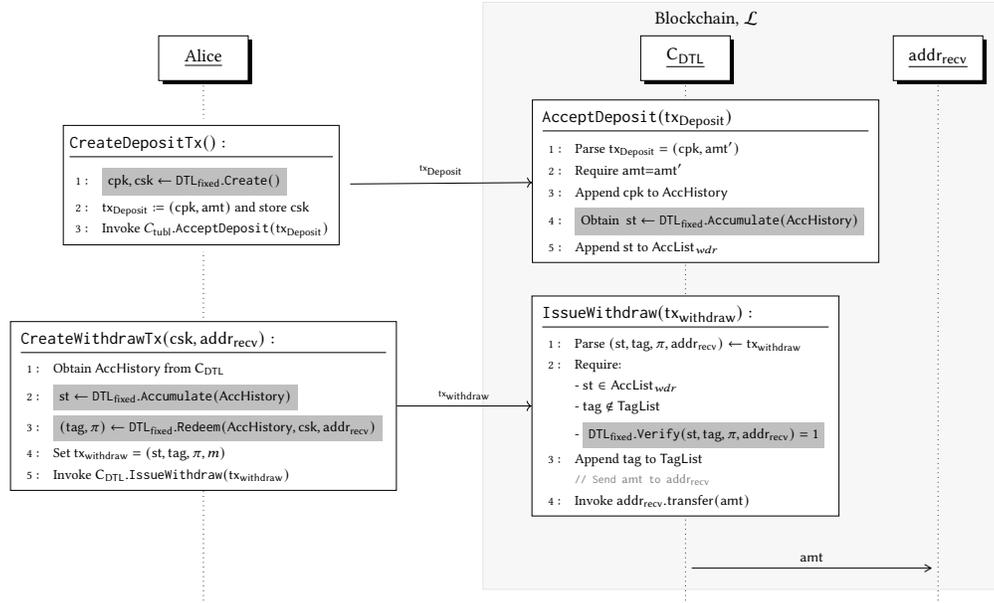

    \centering
    \resizebox{1.6\columnwidth}{!}{
    \begin{sequencediagram}
    \filldraw[fill=black!30, draw=black,  opacity=0.1, name=blockchain] (6,-9.5) rectangle (15.3,1);
    \node at (10, 0.7) {Blockchain, $\mathcal{L}$};
    \newinst{a}{Alice}
    \newinst[7]{m}{$\cdtl$}
    \newinst[2.9]{b}{$\addr_\mathsf{recv}$}
    \IMess{a}{}{m}
    \node (A) [fill=white] at ([yshift=-.95cm] mess from){
    \pcbox{\begin{subprocedure}
    \procedure[linenumbering, bodylinesep=2pt]{$\CreateDTX():$}
    {
      \colorbox{lightgray}{\strut $\dcmcoinpk, \dcmcoinsk \leftarrow \fDDM.\dcmcreatecoin()$}\\
      \txd \define (\dcmcoinpk, \xcoin) \text{ and store } \dcmcoinsk\\
      \text{Invoke } C_{\mathsf{tubl}}.\acceptdeposit(\txd)
    }
    \end{subprocedure}}
    };
    \node (B) [fill=white, draw=black] at ([xshift=.36cm, yshift=-.88cm] mess to){
    \begin{subprocedure}
      \procedure[linenumbering, bodylinesep=2pt]{$\acceptdeposit(\txd)$}{
        \text{Parse } {\txd = (\dcmcoinpk, \xcoin')}\\
        \text{Require } \xcoin {=} \xcoin'\\ 
        \text{Append } \dcmcoinpk \text{ to } \dpl\\
        \colorbox{lightgray}{\text{Obtain } \strut $\dcmcoinaccstate \leftarrow \fDDM.\dcmaccumulatecoin(\dpl)$}\\
        \text{Append } \dcmcoinaccstate \text{ to } \drootlist
      }
    \end{subprocedure}
    };
   
    \postlevel
    \postlevel
    \postlevel
    \postlevel
    \postlevel
    \postlevel
    \postlevel
    \postlevel
    \postlevel
    \postlevel
    \postlevel
    \postlevel
    \node (C) [fill=white, draw=black] at ([yshift=-4.9cm] mess from){
      \begin{subprocedure}
      \procedure[linenumbering, bodylinesep=2pt]{$\CreateWTX(\dcmcoinsk, \addr_\mathsf{recv}):$}
        {
          \text{Obtain } \dpl \text{ from } \cdtl\\
          \colorbox{lightgray}{$\dcmcoinaccstate \leftarrow \fDDM.\dcmaccumulatecoin(\dpl)$}\\
          \colorbox{lightgray}{$(\nulli,\pi) \leftarrow \fDDM.\dcmredeemcoin(\dpl, \dcmcoinsk, \addr_\mathsf{recv})$}\\
          \text{Set } \txw = (\dcmcoinaccstate, \nulli, \pi, m)\\
          \text{Invoke } \cdtl.\issuewithdraw(\txw)
        }
      \end{subprocedure}
    };
    \node (D) [fill=white, draw=black] at ([yshift=-4.9cm] mess to){
    \begin{subprocedure}
      \procedure[linenumbering, bodylinesep=2pt]{$\issuewithdraw(\txw):$}
      {
        \text{Parse } (\dcmcoinaccstate, \nulli, \pi, \addr_\mathsf{recv}) \leftarrow \txw\\
        \text{Require: } \pcskipln\\
        \text{- } \dcmcoinaccstate \in \drootlist \pcskipln\\
        \text{- } \nulli \notin \wnl \pcskipln\\
        \text{- } \colorbox{lightgray}{$\fDDM.\dcmverify(\dcmcoinaccstate, \nulli, \pi, \addr_\mathsf{recv}) = 1$}\\
        \text{Append } \nulli \text{ to } \wnl \pcskipln\\
        \glcomment{\color{gray} Send \amt to $\addr_\mathsf{recv}$}\\
        \text{Invoke } \addr_\mathsf{recv}.\mathsf{transfer}(\amt)
      } 
    \end{subprocedure}
    };
    \draw [->] (A) -- node [midway,above] {\tiny  $\txd$} (B);
    \draw [->] (C) -- node [midway,above] {\tiny  $\txw$} (D);
    \Mess{m}{$\amt$}{b}
    \end{sequencediagram}
    }
    \vspace{-0.3cm}
    \caption{Fixed-amount Unlinkable Payment Using $\fDDM$ (highlighted in gray). 
    In this construction, Alice uses different addresses to invoke the deposit and withdraw functions. 
    The recipient of the
    funds is encoded as the message during the withdrawal process.
    }
    \label{fig:dml-mixer}
\end{figure*}

%% file: applications.tex
\subsection{Unlinkable Fixed-amount Payment} 
\label{subsec:fixed-amount}

\pparagraph{Problem Statement} 
Assume a set of senders $\mathcal{S}$ and a set of receivers $\mathcal{R}$. 
Each sender $s \in \mathcal{S}$ owns a coin $c_i$ of a fixed value $\amt$. 
In this setting, the unlinkable payment problem consists of transferring each
coin $c_i$ from a sender $s_i$ to a receiver $r_j$ so that the following
properties are maintained: 
\emph{(i) correctness}: the coin $c_i$ is transferred from $s_i$ to the intended $r_j$; 
\emph{(ii) unlinkability}: on input the set of transfers, an adversary cannot determine which honest receiver $r_j$ received the coin from an honest sender $s_i$ better than guessing among the set of honest receivers;
\emph{(iii) theft prevention}: an adversary cannot transfer a coin from honest sender $s_i$ to a receiver other than the intended $r_j$;
\emph{(iv) availability}: an adversary cannot prevent any honest user from sending/receiving their coins.

\pparagraph{Our Solution} 
To implement this application, we use $\fDDM$ 
within the tumbler smart contract, as shown in~\cref{fig:dml-mixer}. 
%
 The system consists of two components: the
 \emph{clients} and the \emph{smart contract} running at address
 $\cdtl$. The clients utilize blockchain addresses to communicate
 with $\cdtl$, which manages an \emph{asset pool}. Clients have the capability
 to either deposit or withdraw coins into or
 from this pool via the following two functionalities: $\CreateDTX$ and
 $\CreateWTX$. 

 Similarly, the smart contract has two functions that help process deposit and
 withdraw transactions from clients: \acceptdeposit and $\issuewithdraw$.
 The system utilizes the $\fDDM$ construction (see~\cref{sec:dcm-construction-simple}) and works as follows:
 \begin{enumerate}[leftmargin=0.5cm]
   \item \emph{Deposit, $\mathsf{addr_{client}} \rightarrow \cdtl$:} 
     Upon receiving a valid deposit from a client, $\cdtl$ accumulates it into a Merkle tree. 
 \item \emph{Unlinkable Withdrawal, $\cdtl \rightarrow \mathsf{addr_{receiver}}$:}
 In order to make a payment, clients issue
 a withdrawal transaction with message $m$ designating the recipient address. This transaction demonstrates that the client possesses
 the witness to one of the public keys accumulated to the Merkle tree. If
 verified as valid, the smart contract transfers the funds to the recipient.
 \end{enumerate}


 
 Given a secure $\fDDM$, reasoning about the
 security properties of the suggested construction becomes straightforward. 
 Unlinkable Fixed-amount Payment defined in \cref{fig:dml-mixer} effectively ensures \emph{correctness},
 \emph{unlinkability}, and \emph{theft prevention}, primarily due to the inherent
 properties of the \fDDM. Regarding \emph{availability}, it is ensured by the
 censorship resistance provided by the underlying blockchain protocol and the non-slanderability property of $\fDDM$.

\begin{figure*}[!h]
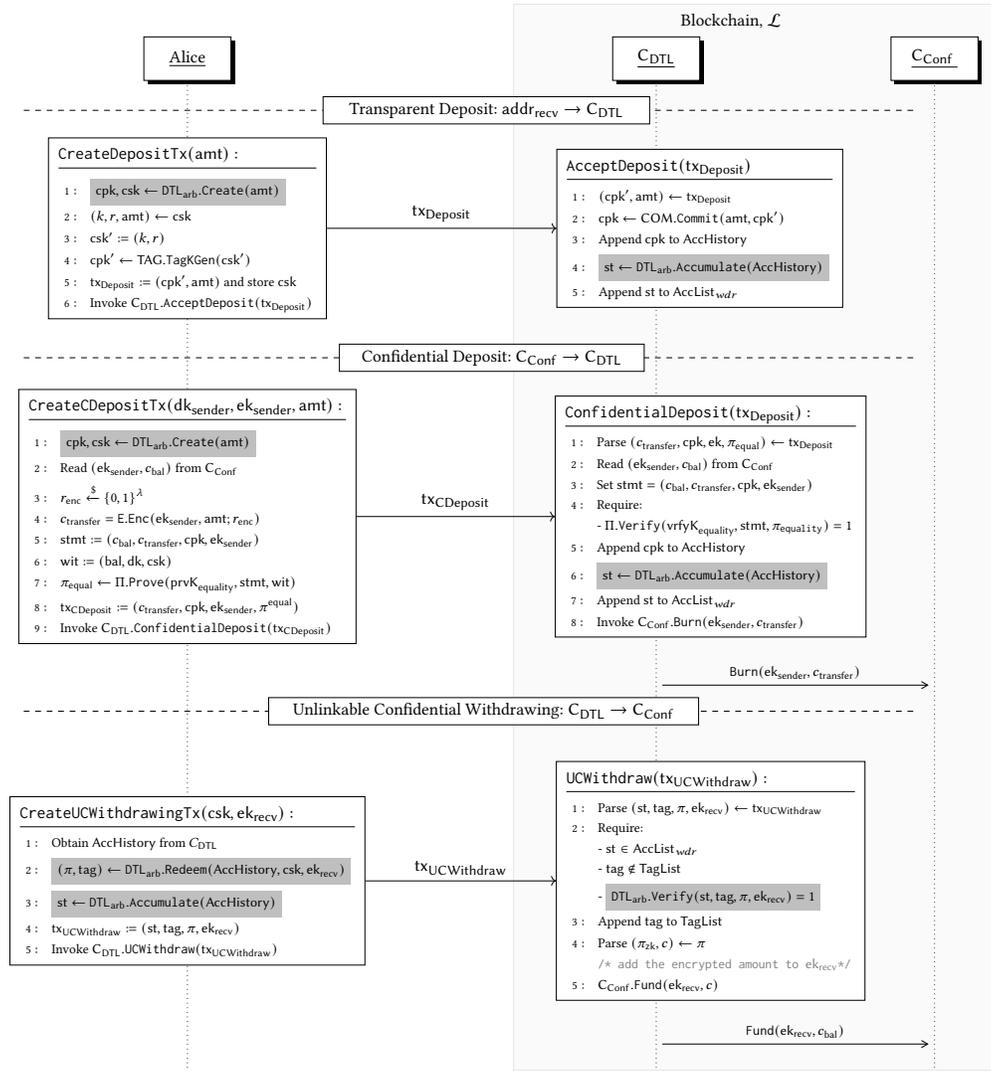

    \centering
    \resizebox{1.6\columnwidth}{!}{
    \begin{sequencediagram}
    \filldraw[fill=black!20, draw=black,  opacity=0.1, name=blockchain] (7,-18.6) rectangle (16,1);
    \node at (11, 0.7) {Blockchain, $\mathcal{L}$};
    \newinst{a}{Alice}
    \newinst[7]{m}{$\cdtl$}
    \newinst[3.5]{b}{$\ccon$ }
    \postlevel
    \draw[dashed, line width=0.2pt] (-2,-0.95) -- (14.4,-0.95);
    \filldraw[fill=white] (3.5,-0.7) rectangle (9.5,-1.2) node[midway] {\normalsize {Transparent} Deposit: $ \addr_{\mathsf{recv}}\rightarrow \cdtl$};
    \postlevel
    \postlevel
    \postlevel
    \postlevel
    \postlevel
    \IMess{a}{}{m}
    \node (A) [fill=white, draw=black] at ([yshift=1.8cm] mess from){
    \begin{subprocedure}
    \procedure[linenumbering, bodylinesep=2pt]{$\CreateDTX(\amt):$}
    {
      \colorbox{lightgray}{\strut $\dcmcoinpk, \dcmcoinsk \leftarrow \aDDM.\dcmcreatecoin(\amt)$}\\
      (k,r, \amt) \leftarrow \dcmcoinsk\\
      \dcmcoinsk' \define (k,r)\\
      \dcmcoinpk' \leftarrow \tagscheme.\tagkeygen(\dcmcoinsk')\\
      \txd \define (\dcmcoinpk', \amt) \text{ and store } \dcmcoinsk\\
      \text{Invoke } \cdtl.\acceptdeposit(\txd)
    }
    \end{subprocedure}
    };
    \node (B) [fill=white, draw=black] at ([xshift=0.8cm, yshift=1.8cm] mess to){
    \begin{subprocedure}
       \procedure[linenumbering, bodylinesep=2pt]{$\acceptdeposit(\txd)$}{
        (\dcmcoinpk', \amt) \leftarrow \txd\\
        \dcmcoinpk \leftarrow \commit.\commitcreate(\amt, \dcmcoinpk')\\
        \text{Append } \dcmcoinpk \text{ to } \dpl\\
        \colorbox{lightgray}{\strut $\dcmcoinaccstate \leftarrow \aDDM.\dcmaccumulatecoin(\dpl)$}\\
        \text{Append } \dcmcoinaccstate \text{ to } \drootlist
      }
    \end{subprocedure}
    };
    \draw [->] (A) -- node [midway,above] {$\txd$} (B);
    \draw[dashed, line width=0.2pt] (-2,-5.5) -- (14.4,-5.5);
    \filldraw[fill=white] (3.8,-5.25) rectangle (9.4,-5.75) node[midway] {\normalsize\normalsize {Confidential} Deposit: $\ccon\rightarrow \cdtl$};
    \node (C) [fill=white, draw=black] at ([yshift=-3.5cm] mess from){
    \begin{subprocedure}
      \procedure[linenumbering, bodylinesep=2pt]{$\CreateCFTX(\dk_\mathsf{sender}, \ek_\mathsf{sender}, \amt):$}
      {
        \colorbox{lightgray}{\strut $\dcmcoinpk, \dcmcoinsk \leftarrow \aDDM.\dcmcreatecoin(\amt)$}\\
        {\text{Read } (\ek_{\mathsf{sender}}, c_{\mathsf{bal}}) \text{ from } \ccon}\\
        r_{\mathsf{enc}} \sample \bset^\lambda\\
        c_{\mathsf{transfer}}=\mathtt{E}.\mathtt{Enc}(\ek_{\mathsf{sender}}, \amt; r_{\mathsf{enc}})\\
        \statement \define (c_{\mathsf{bal}}, c_{\mathsf{transfer}}, \dcmcoinpk, \ek_{\mathsf{sender}})\\
        \witness \define (\bal, \dk, \dcmcoinsk)\\
        \pi_{\mathsf{equal}} \leftarrow \zkprove(\pk_{\mathsf{equality}}, \statement, \witness)\\
        \txcd \define (c_{\mathsf{transfer}}, \dcmcoinpk, \ek_{\mathsf{sender}}, \pi^{\mathsf{equal}})\\
        \text{Invoke }\cdtl.\acceptconfideposit(\txcd)
      }
    \end{subprocedure}
    };
    \node(D) [fill=white, draw=black] at ([xshift=1cm, yshift=-3.5cm] mess to){
    \begin{subprocedure}
    \procedure[linenumbering, bodylinesep=2pt]{$\acceptconfideposit(\txd):$}
      {
        \text{Parse } (c_{\mathsf{transfer}}, \dcmcoinpk, \ek, \pi_{\mathsf{equal}}) \leftarrow \txd\\
        \text{Read } (\ek_{\mathsf{sender}}, c_{\mathsf{bal}}) \text{ from } \ccon\\
        \text{Set } \statement=(c_{\mathsf{bal}}, c_{\mathsf{transfer}}, \dcmcoinpk, \ek_{\mathsf{sender}})\\
        \text{Require: } \pcskipln\\
        {\text{- } \zkverify(\vk_{\mathsf{equality}}, \statement, \pi_{\texttt{equality}}) = 1}\\
        \text{Append } \dcmcoinpk \text{ to } \dpl\\
        \colorbox{lightgray}{\strut $\dcmcoinaccstate \leftarrow \aDDM.\dcmaccumulatecoin(\dpl)$}\\
        \text{Append } \dcmcoinaccstate \text{ to } \drootlist\\
        \text{Invoke } \ccon.\mathtt{Burn}(\ek_\mathsf{sender}, c_{\mathsf{transfer}})
      } 
    \end{subprocedure}
    };
    \draw [->] (C) -- node [midway,above] {$\txcd$} (D);
    \postlevel
    \postlevel
    \postlevel
    \postlevel
    \postlevel
    \postlevel
    \postlevel
    \postlevel
    \postlevel
    \postlevel
    \Mess{m}{$\mathtt{Burn}(\ek_{\mathsf{sender}}, c_{\mathsf{transfer}})$}{b}
    \IMess{a}{}{m}
    \draw[dashed, line width=0.2pt] (-2,-12) -- (14.4,-12);
    \filldraw[fill=white] (2.5,-11.75) rectangle (10.4,-12.25) node[midway] {\normalsize Unlinkable Confidential Withdrawing: $\cdtl\rightarrow \ccon$};
    \node (E) [fill=white, draw=black] at ([yshift=-3cm] mess from){ 
      \begin{subprocedure}
        \procedure[linenumbering, bodylinesep=2pt]{$\CreateCFWTX(\dcmcoinsk,\ek_\mathsf{recv}):$}
        {
          \text{Obtain } \dpl \text{ from } C_\mathsf{DTL}\\
          \colorbox{lightgray}{\strut $(\pi, \nulli) \leftarrow \aDDM.\dcmredeemcoin(\dpl, \dcmcoinsk, \ek_\mathsf{recv})$}\\
          \colorbox{lightgray}{\strut $\dcmcoinaccstate \leftarrow \aDDM.\dcmaccumulatecoin(\dpl)$}\\
          \txcw \define (\dcmcoinaccstate, \nulli, \pi, \ek_\mathsf{recv})\\
          \text{Invoke }\cdtl.\issueconfiwithdraw(\txcw)
        }
      \end{subprocedure}
      };
    \node (F) [fill=white, draw=black] at ([xshift=1cm, yshift=-3cm] mess to){
    \begin{subprocedure}
      \procedure[linenumbering, bodylinesep=2pt]{$\issueconfiwithdraw(\txcw):$}
      {
        \text{Parse } (\dcmcoinaccstate, \nulli, \pi, \ek_\mathsf{recv}) \leftarrow \txcw\\
        \text{Require: } \pcskipln\\
        \text{- } \dcmcoinaccstate \in \drootlist \pcskipln\\
        \text{- } \nulli \notin \wnl \pcskipln\\
        \text{- } \colorbox{lightgray}{\strut $\aDDM.\dcmverify(\dcmcoinaccstate, \nulli, \pi, \ek_\mathsf{recv}) = 1$}\\
        \text{Append } \nulli \text{ to } \wnl \\
        \text{Parse } (\zkproof, c) \leftarrow \pi \pcskipln\\ 
        \texttt{\color{gray}/* add the encrypted amount to $\ek_{\mathsf{recv}}$*/}\\
        \ccon.\fun{Fund}(\ek_\mathsf{recv}, c)
      } 
    \end{subprocedure}
    };
    \draw [->] (E) -- node [midway,above] {$\txcw$} (F);
    \postlevel
    \postlevel
    \postlevel
    \postlevel
    \postlevel
    \postlevel
    \postlevel
    \postlevel
    \postlevel
    \Mess{m}{$\mathtt{Fund}(\ek_{\mathsf{recv}}, c_{\bal})$}{b}
    
    \end{sequencediagram}
    }
    \vspace{-0.3cm}
    \caption{Confidential and Unlinkable Payment Using~\aDDM~(highlighted in gray). 
      }
    \label{fig:ddm-zether-plaintext-deposit}
\end{figure*}

\subsection{Unlinkable Confidential Payment}
\label{sub:ddmzether}

The approach in section~\ref{subsec:fixed-amount} requires  
transactions with fixed amounts. 
Here we show how to integrate $\aDDM$ with confidential payment contracts like
Zether~\cite{zether-bunz-2020}, to provide their users with additional unlinkability guarantees.   
This leads to a payment system with 
unlinkability and confidentiality.

\pparagraph{Problem Statement} We consider here the same problem statement as
in~\ref{subsec:fixed-amount}, relaxing the condition that all payments must be
of a fixed value.

\pparagraph{Our Solution}
To implement this application, we start with the assumption of a confidential payment system like Zether ($\ccon$). We then compose this system with a smart contract, $\cdtl$,
using  $\aDDM$. 
The specifics of our approach are illustrated in \cref{fig:ddm-zether-plaintext-deposit}.
In this scenario, $\cdtl$ is assumed to have the capability to directly fund and burn the balances of 
 accounts on $\ccon$ through the functions $\mathtt{Fund}()$ and $\mathtt{Burn}()$ upon valid transactions.
In this construction, \sloppy{$\CreateDTX$, $\CreateCFTX$, and $\CreateCFWTX$ outline the methods which clients create transactions to deposit} \emph{transparent} amounts from the blockchain to $\cdtl$, 
deposit \emph{confidential} amounts from $\ccon$ to the $\cdtl$, 
and fund confidential amounts from $\cdtl$ to $\ccon$, respectively.
\begin{enumerate}[leftmargin=0.5cm]
  \item \emph{Transparent Deposit, $\addr_{\mathsf{sender}}\rightarrow \cdtl$.} 
  Similar to the deposit procedure described in \cref{subsec:fixed-amount}, this method enables existing users to deposit arbitrary amounts of tokens from the blockchain into $\cdtl$, preparing them for subsequent withdrawal into a confidential contract like $\ccon$. Upon receiving a deposit, $\cdtl$ integrates the depositor's public key into the Merkle Tree.
  
\item \emph{Confidential Deposit, $\ccon \rightarrow \cdtl$.}
    In this scenario, clients aim to deposit a confidential token amount, denoted as $c_{\mathsf{transfer}} = \mathtt{Enc}(\ek, \amt)$, from $\ccon$ to the $\cdtl$ contract. Upon receiving a valid confidential deposit, $\cdtl$ stores $\dcmcoinpk$, a commitment to the same amount $\commitcreate(\amt, \cdot)$, as a leaf in the Merkle tree. 
    However, before updating the tree, the $\cdtl$ contract must verify that $c_{\mathsf{transfer}}$ uses the identical $\amt$ value used in $\dcmcoinpk$. 
    To facilitate this, we introduce an additional zero-knowledge proof here. 
    We denote $c_{\mathsf{bal}}$ as the sender's encrypted balance.
    The concrete relation is as in~\cref{eq:confideposit}.

\item \textit{Unlinkable Confidential Withdrawal, $\cdtl \rightarrow \ccon$.} 
    The unlinkable confidential withdrawal is straightforward and directly leverages \dcmredeemcoin in \aDDM. 
    Senders obtain a valid redeem proof containing $c_{\mathsf{transfer}}$, an encryption of $\amt$ using the recipient's encryption key. This ciphertext can be homomorphically added to the recipient's encrypted balance to update it.
\end{enumerate}
{\footnotesize 
    \begin{equation}~\label{eq:confideposit}
            R_{\mathsf{equality}} \define \\
            \left\{ 
            \begin{aligned}
            & \statement \define (c_{\mathsf{bal}}, c_{\mathsf{transfer}}, \dcmcoinpk, \ek);\\
            & \witness \define  (\mathsf{bal}, \dk, \dcmcoinsk = (k, r, \amt)): \\
            & \mathsf{bal} = \fun{Dec}(\dk, c_{\mathsf{bal}})\wedge
            \amt = \fun{Dec}(\dk, c_{\mathsf{transfer}}) \wedge\\
            &  \ek=\fun{derive}(\dk)\wedge 0 \leq \amt \leq \mathsf{bal}~\wedge \\
            & \dcmcoinpk' = \tagscheme.\tagkeygen(\dcmcoinsk) \land \\
            & \dcmcoinpk = \commit.\commitcreate(\amt; \dcmcoinpk')
            \end{aligned}
          \right\}
    \end{equation}
}
Given a secure $\aDDM$, it is straightforward to reason that our unlinkable and confidential payment achieves similar properties as defined in \cref{subsec:fixed-amount}.

\pparagraph{Flexibility in Choosing Privacy Guarantee} {The composability of $\cdtl$ grants users the flexibility to choose between issuing unlinkable, confidential, or both types of transactions, without requiring that all transactions adhere to the unlinkable model. This level of flexibility is not commonly offered by existing solutions like Monero, Zcash, or Zether, which typically focus on either unlinkability or confidentiality.

{Users can prioritize privacy or efficiency based on their needs. An unlinkable and confidential payment can take two forms: Transparent Deposit (1) followed by Unlinkable Confidential Withdrawal (3), or Confidential Deposit (2) followed by Unlinkable Confidential Withdrawal (3).  This type of payment, detailed in~\cref{sec:evaluation}, consumes around $1.1$ to $1.8$ million gas. Alternatively, a confidential payment can be executed entirely within $\ccon$, requiring only one step and costing approximately $1$ million gas, as shown in \cite{eip1108}. }


}

\input{protocols/anon-zether}
\subsection{Unlinkable Weighted Voting}
\label{sub:ddmvoting}

We further highlight the composability property of DTL by demonstrating its ability to constuct unlinkable and weighted voting protocols. This effectively addresses the shortcomings, such as the absence of privacy and confidentiality, commonly found in existing DAO (Decentralized Autonomous Organization) voting protocols. 

\pparagraph{Problem Statement}
We aim to create a voting system that incorporates varying voting powers among participants. In contrast to typical blockchain voting systems, our objective is to additionally achieve the property of \emph{Ballot Privacy}. This entails safeguarding the confidentiality of individual voters' preferences, encompassing both the candidate they vote for and the quantity of their vote.

\pparagraph{Our Solution} Similar to the construction outlined in section
\ref{sub:ddmzether}, the protocol can be implemented using a contract,
$\cdtl$ using $\aDDM$ for unlinkability and a contract
$\ccon$ for confidentiality. For a voting protocol, our system can
be divided into four steps: \emph{setup, registration, voting}, and
\emph{reveal}. In particular, these phases are: 

\begin{enumerate}[leftmargin=0.5cm]
\item \emph{Setup Phase, $\addr_{i} \rightarrow
  \ccon$:} In the \emph{Setup} phase, candidates, $\addr_{candidate}$, register their
  public encryption key to $\ccon$. Conceptually, this can be understood as
  each candidate being associated with a unique encryption key to which voters
  can securely send their votes.
\item \emph{Registration Phase, $\addr_{\mathsf{voter}} \rightarrow \cdtl$:}
 The \emph{Registration} phase resembles the transparent deposit, where voters deposit tokens into $\cdtl$. 
 These tokens act as representations of voters' eligibility and voting influence. 
 They should be allocated to voters before the start of the voting process.

\item \emph{Unlinkable and Confidential Voting, $\cdtl \rightarrow \ccon$}: 
This process mirrors the unlinkable and confidential withdrawing described in~\cref{sub:ddmzether}. Once all participants have completed registration, the voting phase starts. Voters cast their ballots for candidates with both unlinkability and confidentiality ensured. Voters achieve this by withdrawing voting powers from $\cdtl$ and allocating them to their chosen candidates, encrypting the number of tokens in the process. Leveraging the additive homomorphic properties of the encryption scheme, the sum of all encrypted votes for a candidate can be securely computed, resulting in the encryption of the total votes received by that candidate. Utilizing the capabilities of $\aDDM$ effectively eliminates the link between voters and their selected candidates.

\item \emph{Reveal, $\addr_{i} \rightarrow \ccon$}:
In the \emph{reveal} stage, each candidate, $\addr_{i}$, discloses their total vote count by revealing the decrypted total of their vote balance. To ensure the integrity of this disclosure, candidates must prove that the revealed total is the actual decrypted amount. Therefore, the candidate is required to additionally provide a zero-knowledge proof to validate the following statement:

{\footnotesize 
\begin{equation*}
    \begin{split}
        R_{\mathsf{reveal}} := \left\{
        \begin{aligned}
        &\statement\define (\ek_i, c_i, \bal_i);
        \witness \define \dk : \\
        & \ek_i = \fun{Derive}(\dk)\wedge  \bal_i = \fun{E.Dec}(\dk, c_i)
        \end{aligned}
      \right\}
    \end{split}
\end{equation*}
}

where $c$ and $bal$, respectively, stand for the encrypted and plaintext accumulated votes. The prover proves it knows the decryption key $\dk$ corresponds to a public encryption key $\ek$, and all the decryptions are correct.
\end{enumerate}





In this protocol, it is straightforward to see that ballot privacy is guaranteed by properties of $\aDDM$.

Finally, due to the space constraints, we move the details of this protocol to~\cref{sec:voting}.

%% file: evaluation.tex
\section{Evaluation}
\label{sec:evaluation}

\subsection{Testbed}
\pparagraph{Cryptographic Primitives} 
In this evaluation, we use Groth16 \zksnark~\cite{groth2016size} as our zero-knowledge proof system
for its efficiency in proof size and on-chain verification cost.
We use the Poseidon Hash~\cite{grassi2021poseidon} to instantiate Merkle Tree as it yields a fewer number of
constraints in circuits (or R1CS), and a lower gas cost on blockchain. Similarly,
we use the Pedersen Hash~\cite{pederson-hash} for the implementation of the tagging scheme 
as well as the commitment scheme. 
In \aDDM, we use ElGamal encryption with the Baby Jubjub
elliptic curve. The underlying prime field of Baby Jubjub has the same
order as the BN254 curve. This is essential because this
makes it easier to do elliptic curve operations inside a circuit whose
satisfiability is going to be proved by Groth16 (with BN254). 
 For all applications, we set the Merkle tree depth to $20$, allowing for an anonymity set of up to $2^{20}$. 
 
\pparagraph{Software and  Hardware} We implemented all the circuits (R1CS) for
the zero-knowledge proof using \texttt{Circom}~\cite{circomlib}. We then use the \texttt{SnarkJS}~\cite{snarkjs} library for proof generation and on-chain verifier creation in Solidity. We deploy
it to a test network using the Truffle toolchain. We test our implementation on
a PC with 4-core 11th Gen Intel(R) Core(TM) i5-1137G7 @2.40Hz and 8GB memory.


\subsection{Performance}
We implemented all proposed applications and report the number of constraints in R1CS, proving time, and verification gas cost.

\pparagraph{Unlinkable Fixed-amount Payment} 
We report the numbers for the unlinkable fixed-amount payment in~\cref{tab:fixedPayment},
where $\addr_{\mathsf{sender}}$, $\cdtl$, and $\addr_{\mathsf{receiver}}$ 
denotes the sender, the tumbling contract employing $\fDDM$, and the receiver, respectively.
We recall that an unlinkable payment works as follows: A sender can deposit
money to the contract and later withdraw it to the recipient address ($\addr_{\mathsf{sender}}\rightarrow \cdtl \rightarrow \addr_{\mathsf{receiver}}$).
\begin{table}[t]
\centering
\caption{Unlinkable Fixed-amount Payment Performance.}
\resizebox{.8\columnwidth}{!}{%
\begin{tabular}{@{}ccc@{}}
\toprule
                 & \textbf{\begin{tabular}[c]{@{}c@{}} Deposit \\ $\addr_{\mathsf{sender}} \rightarrow \cdtl$ \end{tabular}} & \textbf{\begin{tabular}[c]{@{}c@{}} Withdraw \\ $\cdtl \rightarrow \addr_{\mathsf{receiver}}$ \end{tabular}}\\ \midrule
    R1CS Constraints & -                 & $8,146$             \\     
    Proving Time     & -                 & $1.15$s             \\     
    Gas Cost         & $767,565$         & $233,375$         \\     \bottomrule 
\end{tabular}
}
\label{tab:fixedPayment}
\end{table}

\pparagraph{Unlinkable and Confidential Payment}
We report the numbers for the unlinkable and confidential payment in Table
\ref{tab:anonPayment}, where $\addr_{\mathsf{client}}$, $\cdtl$, and $\ccon$
stand for the address of the client, the DTL Contract, and the Confidential Payment Contract (i.e.,
Zether), respectively. 
We recall that there are two ways to conduct unlinkable and confidential transfers: 
a user can send money from $\addr_{\mathsf{client}}$ to $\cdtl$, then withdraw to $\ccon$ (i.e., $\addr_{\mathsf{client}} \rightarrow \cdtl \rightarrow \ccon$), 
or user can transfer money from  $\ccon$ to $\cdtl$ then withdraw back to $\ccon$ ($\ccon \rightarrow \cdtl \rightarrow \ccon$). 
This unlinkable and confidential transfer incurs a total gas cost of approximately $1m$ for the first approach and $1.8m$ for the second approach, for an anonymity set size of up to $2^{20}$.
This is 
a significant improvement upon the previous state-of-the-art anonymous and
confidential transfer solution~\cite{diamond2021many}, where the gas cost is
$7.3m$ for an anonymity set of $8$ users, rising to $36m$ for $64$ users, exceeding the Ethereum block gas limit of $30m$\footnote{\url{https://ethereum.org/en/developers/docs/gas/}}.

\begin{table*}[t]
\centering
\caption{Confidential and Unlinkable Payment Performance. The anonymity set of Anonymous Zether transfer reported in the table is $8$, while ours is $2^{20}$.}
\resizebox{1.4\columnwidth}{!}{%
\begin{tabular}{@{}ccccc@{}}
\toprule
                 & \textbf{\begin{tabular}[c]{@{}c@{}} Transparent Deposit \\ $\addr_{\mathsf{client}} \rightarrow \cdtl$ \end{tabular}} & \textbf{\begin{tabular}[c]{@{}c@{}} Confidential Deposit \\ $\ccon\rightarrow \cdtl$ \end{tabular}} & \textbf{\begin{tabular}[c]{@{}c@{}}Unlinkable Withdraw\\ $\cdtl\rightarrow \ccon$\end{tabular}} & \begin{tabular}[c]{@{}c@{}}\textbf{AZether}~\cite{diamond2021many} \\ Transfer ($8$)\end{tabular} \\ \midrule
    R1CS Constraints & -                 & $8,671$             & $11,521$  &   -        \\     
    Proving Time     & -                 & $1.16$s             & $1.32$s   &   $1.9$s   \\     
    Gas Cost         & $767,565$         & $1,051,112$         & $258,467$ &   $7,306,703$         \\     \bottomrule 
\end{tabular}
}
\label{tab:anonPayment}
\end{table*}

\pparagraph{Unlinkable Weighted Voting}
Recall that there are three main phases in unlinkable weighted voting: the registration phase, the voting phase, and the revealing phase. We omit the setup phase as it is fairly straightforward. In this context, $\addr_{\mathsf{voter}}$, $\addr_{i}$, $C_\mathsf{DTL}$, and $\ccon$ represent the voter, the candidate, the DTL contract, and the confidential contract, respectively.
Voters register their tokens first and then vote for the candidate with the token in an unlinkable manner. 
At the end of voting, the candidate will reveal all the votes in zero knowledge with the private
key. 
The registration phase is akin to the deposit step in other applications, primarily due to the Merkle tree cost. Similarly, the unlinkable voting phase mirrors that of unlinkable withdrawal in confidential and unlinkable payments. The only difference is the reveal step 
performed by the candidate. 
The numbers are reported in \cref{tab:anonVote}.

\begin{table}[t]
\centering
\caption{Unlinkable Weighted Voting Cost}
\resizebox{.9\columnwidth}{!}{%
\begin{tabular}{@{}cccc@{}}
\toprule
                 & \textbf{\begin{tabular}[c]{@{}c@{}}Registration\\ $\addr_{\mathsf{voter}} \rightarrow \cdtl$ \end{tabular}}  & \textbf{\begin{tabular}[c]{@{}c@{}} Unlinkable Vote \\ $\cdtl\rightarrow \ccon$ \end{tabular}} &  \textbf{\begin{tabular}[c]{@{}c@{}}Reveal\\ $\addr_{i}\rightarrow \ccon$\end{tabular}}    \\ \midrule
R1CS Constraints & -    &  $11,521$ & $6,750$      \\
Proving Time     & -     & $1.19$s   & $1.03$s        \\
Gas Cost         & $767,565$ & $258,467$ & $291,151$  \\ \bottomrule
\end{tabular}
}
\label{tab:anonVote}
\end{table}

\section{Discussion}
\label{sec:discussion}
\pparagraph{\zksnark setup} As discussed in Section \ref{sec:evaluation}, we used Groth16 \zksnark to instantiate DTL for efficiency, which requires a trusted setup~\cite{groth2016size}. However, our DTL framework is \zksnark agnostic and can be instantiated with any \zksnark proof system providing the API in~\cref{sec:prelim}. If the trusted setup assumption is undesirable, DTL can be instantiated using proof systems utilizing a multi-party computation (MPC)
protocol where multiple users contribute shares to the trusted setup~\cite{SP:BCGTV15,FCW:BowGabGre18,EPRINT:BowGabMie17}, or \zksnark proof systems with a universal setup~\cite{CCS:CamFioQue19,EC:CHMMVW20,EPRINT:GabWilCio19}, though at an efficiency cost.

\pparagraph{Network Layer Privacy}
Most blockchain wallet clients, like MetaMask, typically rely on centralized services such
as Infura for blockchain data. These services, having access to the clients'
blockchain addresses, IP addresses, and queries about contract states,
hold significant privacy implications.
They could potentially
link various addresses belonging to the same client. 
To enhance privacy, it is
advisable for clients to either run an independent validating full blockchain
node or employ network-level anonymity tools, such as Tor, Nym~\cite{Daz2021TheNN}, or a Virtual Private
Network (VPN), prior to connecting with these centralized services. This
approach can help in safeguarding their privacy and reducing the risk of
address linkage.

\pparagraph{Using Relayer For Transaction Fee}
Initiating a transaction necessitates the payment of fees. To preserve
anonymity, clients should avoid using the same address for these payments, as
doing so may allow their addresses to be linked. In real-world applications,
users can use a relayer. This relayer can send transactions on behalf of users
and is potentially compensated for a portion of a successful transaction. The address of 
the relayer can be encoded in the application-dependent message, $m$. The
relayer can obtain the relevant client proof via a separate communication
channel.

\pparagraph{Further Optimization: Batch Update of Merkle Tree} A single update
to the Merkle Tree constitutes a significant cost in our applications, coming in at
$767k$ gas on the blockchain. A potential enhancement is to employ the technique
introduced in~\cite{mountain-torn2023}. The contract can postpone the update
process until adequate new leaves are available to formulate a compact
sub-tree. 
update. 
Implementing this method could significantly decrease the gas cost for transactions related to Merkle Tree updates by $0.5m$ gas for a subtree of size $4$. 
\pparagraph{Application to Anonymous Rate Limiting Nullifier} Anonymous Rate Limiting Nullifiers (RLNs)~\cite{ducdske2022, RLNgithub} are cryptographic constructs that enable spam prevention in anonymous messaging systems by limiting each user to one message per time period while preserving their privacy. The key challenge is maintaining unlinkability between messages from the same user across different time periods while ensuring they cannot send multiple messages within the same period.

The Data Tumbling Layer (\fDDM) provides a simple way to implement RLNs. In DTL, users can commit their secret key, which then generates unique nullifiers per round through $\mathtt{TAG}.\mathtt{TagEval}(\mathsf{csk} || \mathsf{round})$. The redeem function can be used to verify membership and check for duplicate nullifiers within rounds, preventing spam while maintaining cross-round unlinkability. This approach leverages $\fDDM$'s proven security properties for membership proofs while adding the rate-limiting capability through round-specific nullifiers, providing a robust foundation for anonymous spam prevention.
Moreover, the system can further enforce a slashing mechanism on spam users by having them commit to a polynomial using polynomial commitment~\cite{kate2010constant} which helps reveal their identity if they exceed the message threshold.

\section{Conclusion}
\label{sec:conclusion}
In this work, we introduced the Data Tumbling Layer (DTL), a novel cryptographic primitive that enables non-interactive data tumbling with strong security and privacy guarantees. Through two constructions -  $\fDDM$ for fixed data and $\aDDM$ for arbitrary data - we demonstrated how to achieve crucial security properties including theft prevention, non-slanderability, unlinkability, and no one-more redemption. We proved the security of both constructions under standard cryptographic assumptions and showcased DTL's practical utility through three concrete applications: unlinkable fixed-amount payments, unlinkable confidential payments, and unlinkable weighted voting. Our implementation and evaluation demonstrated that DTL transactions can be processed efficiently, taking less than 1.5 seconds on a standard laptop with gas costs under 1.8 million, while supporting large anonymity sets of up to $2^{20}$. By providing a modular primitive that can be easily integrated with existing applications, DTL bridges an important gap in blockchain privacy solutions, enabling developers to add unlinkability guarantees to their applications without designing complex cryptographic protocols from scratch.


\section*{Disclaimer}
Case studies, comparisons, statistics, research and recommendations are provided ``AS IS'' and intended for informational purposes only and should not be relied upon for operational, marketing, legal, technical, tax, financial or other advice.  Visa Inc. neither makes any warranty or representation as to the completeness or accuracy of the information within this document, nor assumes any liability or responsibility that may result from reliance on such information.  The Information contained herein is not intended as investment or legal advice, and readers are encouraged to seek the advice of a competent professional where such advice is required.

These materials and best practice recommendations are provided for informational purposes only and should not be relied upon for marketing, legal, regulatory or other advice. Recommended marketing materials should be independently evaluated in light of your specific business needs and any applicable laws and regulations. Visa is not responsible for your use of the marketing materials, best practice recommendations, or other information, including errors of any kind, contained in this document.

%% file: tmp.tex
\newcommand{\dcmoutput}{\mathsf{output}}
\newcommand{\dcminput}{\mathsf{input}}
\newcommand{\generatorsetinput}{\mathsf{setinput}}
\newcommand{\generatorcompute}{\mathsf{compute}}

%% file: security-analysis.tex
\section{Detailed Construction for Unlinkable Weighted Voting}
\label{sec:voting}
This section contains the detailed construction of the unlinkable weighted voting protocol. The detailed description of this protocol is in \cref{fig:voting-voter} and \cref{fig:voting-candidate}.
In both protocols for candidates and voters, we assume there exists a helper function, $\fun{GetStage}()$
that allows participants to learn the current phase of the voting protocol. 
\begin{figure*}[!t]
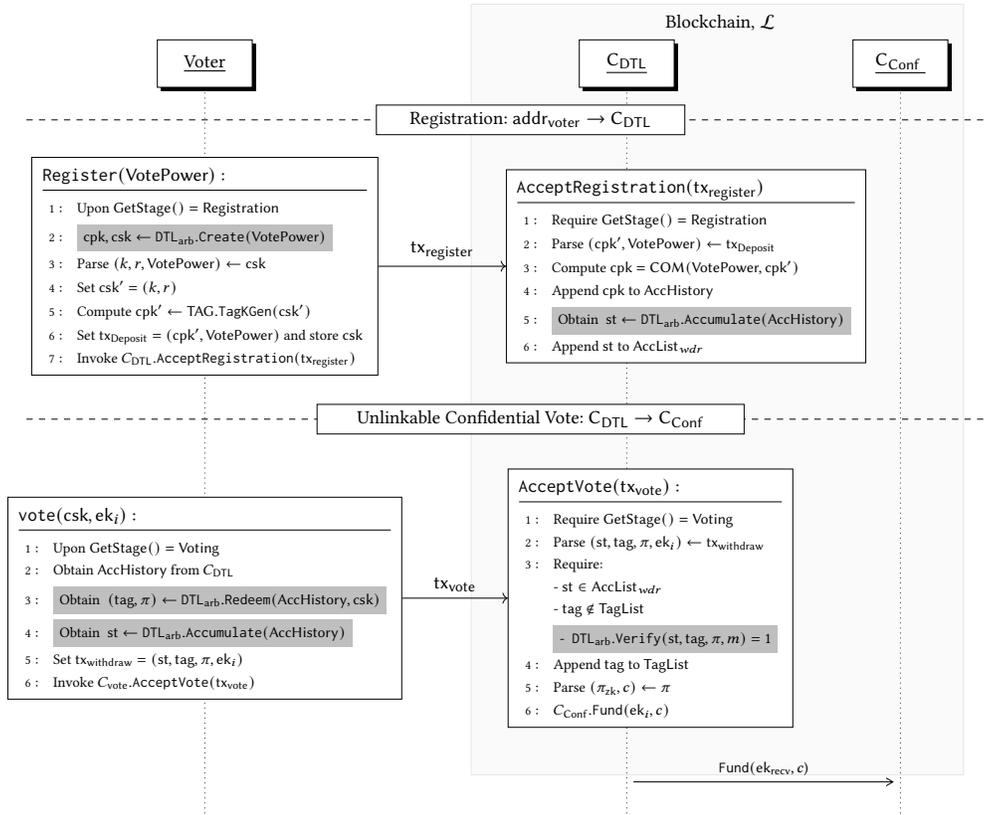

    \centering
    \resizebox{1.6\columnwidth}{!}{
    \begin{sequencediagram}
    \filldraw[fill=black!20, draw=black,  opacity=0.1, name=blockchain] (5.5,-12.) rectangle (13.8,1);
    \node at (9.7, 0.7) {Blockchain, $\mathcal{L}$};
    \newinst{a}{Voter}
    \newinst[5.5]{m}{$\cdtl$}
    \newinst[3]{b}{$\ccon$ }
    \postlevel
    \draw[dashed, line width=0.2pt] (-2,-0.95) -- (14.4,-0.95);
    \filldraw[fill=white] (3.9,-0.7) rectangle (9.1,-1.2) node[midway] {
    Registration: $ \addr_{\mathsf{voter}}\rightarrow \cdtl$};
    \postlevel
    \postlevel
    \postlevel
    \postlevel
    \postlevel
    \IMess{a}{}{m}
    \node (A) [fill=white, draw=black] at ([yshift=1.5cm] mess from){
    \begin{subprocedure}
    \procedure[linenumbering, bodylinesep=2pt]{$\mathtt{Register}(\mathsf{VotePower}):$}
    {
      \text{Upon } \mathsf{GetStage}() = \mathsf{Registration}\\
      \colorbox{lightgray}{$\dcmcoinpk, \dcmcoinsk \leftarrow \aDDM.\dcmcreatecoin(\mathsf{VotePower})$}\\
      \text{Parse } (k,r, \mathsf{VotePower}) \leftarrow \dcmcoinsk \\
      \text{Set } \dcmcoinsk' = (k,r)\\
      \text{Compute } \dcmcoinpk' \leftarrow \tagscheme.\tagkeygen(\dcmcoinsk')\\
      \text{Set } \txd=(\dcmcoinpk', \mathsf{VotePower}) \text{ and store } \dcmcoinsk\\
      \text{Invoke } C_{\mathsf{DTL}}.\mathtt{AcceptRegistration}(\tx_{\mathsf{register}})
    }
    \end{subprocedure}
    };
    \node (B) [fill=white, draw=black] at ([xshift=1cm, yshift=1.5cm] mess to){
    \begin{subprocedure}
      \procedure[linenumbering, bodylinesep=2pt]{$\mathtt{AcceptRegistration}(\tx_{\mathsf{register}})$}{
        \text{Require } \mathsf{GetStage}() =  \mathsf{Registration}\\ 
        \text{Parse } (\dcmcoinpk', \mathsf{VotePower}) \leftarrow \txd\\
        {\text{Compute } \dcmcoinpk = \commit(\mathsf{VotePower}, \dcmcoinpk')}\\
        \text{Append } \dcmcoinpk \text{ to } \dpl\\
        \colorbox{lightgray}{\text{Obtain } $\dcmcoinaccstate \leftarrow \aDDM.\dcmaccumulatecoin(\dpl)$}\\
        \text{Append } \dcmcoinaccstate \text{ to } \drootlist
      }
    \end{subprocedure}
    };
    \draw [->] (A) -- node [midway,above] {$\tx_{\mathsf{register}}$} (B);
    \draw[dashed, line width=0.2pt] (-2.,-6) -- (14.4,-6);
    \filldraw[fill=white] (2.9,-5.75) rectangle (10.1,-6.25) node[midway] {\normalsize\normalsize {Unlinkable Confidential} Vote: $\cdtl\rightarrow \ccon$};
    \IMess{a}{}{m}
    \node (E) [fill=white, draw=black] at ([yshift=-3.5cm] mess from){ 
      \begin{subprocedure}
        \procedure[linenumbering, bodylinesep=2pt]{$\mathtt{vote}(\dcmcoinsk,\ek_i):$}
        {
          \text{Upon } \mathsf{GetStage}() = \mathsf{Voting}\\
          \text{Obtain } \dpl \text{ from } C_{\mathsf{DTL}}\\
          \colorbox{lightgray}{\text{Obtain } $(\nulli, \pi) \leftarrow \aDDM.\dcmredeemcoin(\dpl, \dcmcoinsk)$}\\
          \colorbox{lightgray}{\text{Obtain } $\dcmcoinaccstate \leftarrow \aDDM.\dcmaccumulatecoin(\dpl)$}\\
          \text{Set } \txw = (\dcmcoinaccstate, \nulli, \pi, \ek_{i})\\
          \text{Invoke } C_{\mathsf{vote}}.\mathtt{AcceptVote}(\tx_{\mathsf{vote}})
        }
      \end{subprocedure}
      };
    \node (F) [fill=white, draw=black] at ([xshift=.4cm, yshift=-3.5cm] mess to){
    \begin{subprocedure}
      \procedure[linenumbering, bodylinesep=2pt]{$\mathtt{AcceptVote}(\tx_{\mathsf{vote}}):$}
      {
        \text{Require } \mathsf{GetStage}() = \mathsf{Voting}\\ 
        \text{Parse } (\dcmcoinaccstate, \nulli, \pi, \ek_{i}) \leftarrow \txw\\
        \text{Require: } \pcskipln\\
        \text{- } \dcmcoinaccstate \in \drootlist \pcskipln\\
        \text{- } \nulli \notin \wnl \pcskipln\\
        \colorbox{lightgray}{\text{- } $\aDDM.\dcmverify(\dcmcoinaccstate, \nulli, \pi, m) = 1$}\\
        \text{Append } \nulli \text{ to } \wnl\\
        \text{Parse } (\zkproof, c) \leftarrow \pi \\ 
        C_{\mathsf{Conf}}.\mathtt{Fund}(\ek_{i}, c)
      } 
    \end{subprocedure}
    };
    \draw [->] (E) -- node [midway,above] {$\tx_{\mathsf{vote}}$} (F);
    \postlevel
    \postlevel
    \postlevel
    \postlevel
    \postlevel
    \postlevel
    \postlevel
    \postlevel
    \postlevel
    \postlevel
    \Mess{m}{$\mathtt{Fund}(\ek_{\mathsf{recv}}, c$)}{b}
    
    \end{sequencediagram}
    }
    \vspace{-0.3cm}
    \caption{Unlinkable Weighted Voting Using~\aDDM~(highlighted in gray): Protocol for Voter. 
      }
    \label{fig:voting-voter}
\end{figure*}

\begin{figure*}[!h]
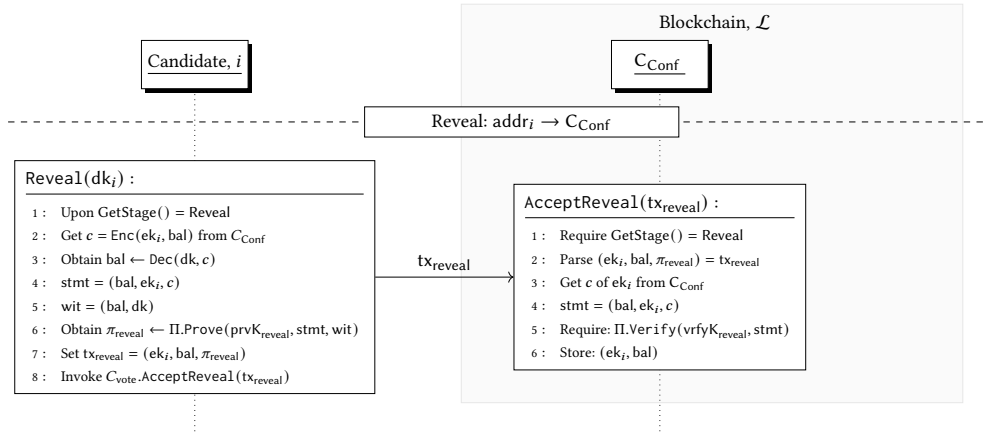

    \centering
    \resizebox{1.6\columnwidth}{!}{
    \begin{sequencediagram}
    \filldraw[fill=black!20, draw=black,  opacity=0.1, name=blockchain] (5.5,-5.6) rectangle (13.8,1);
    \node at (9.7, 0.7) {Blockchain, $\mathcal{L}$};
    \newinst{a}{Candidate, $i$}
    \newinst[6]{b}{$\ccon$ }
    \postlevel
    \draw[dashed, line width=0.2pt] (-2,-0.95) -- (14.4,-0.95);
    \filldraw[fill=white] (3.9,-0.7) rectangle (9.1,-1.2) node[midway] {
    Reveal: $ \addr_{{i}}\rightarrow \ccon$};
    \postlevel
    \postlevel
    \postlevel
    \postlevel
    \postlevel
    \postlevel
    \IMess{a}{}{b}
    \node (A) [fill=white, draw=black] at ([yshift=2cm] mess from){
    \begin{subprocedure}
    \procedure[linenumbering, bodylinesep=2pt]{$\mathtt{Reveal}(\dk_i):$}
    {
      \text{Upon } \mathsf{GetStage}() = \mathsf{Reveal}\\
      \text{Get } c=\mathtt{Enc}(\ek_i, \bal) \text{ from } C_{\mathsf{Conf}}  \\
      \text{Obtain } \bal \leftarrow \mathtt{Dec}(\dk, c)\\
      \statement = (\bal, \ek_i, c)\\
      \witness = (\bal, \dk)\\
      \text{Obtain } \pi_\mathsf{reveal} \leftarrow \zkprove(\pk_{\mathsf{reveal}}, \statement, \witness)\\
      \text{Set } \tx_{\mathsf{reveal}} = (\ek_i, \bal, \pi_{\mathsf{reveal}})\\
      \text{Invoke } C_{\mathsf{vote}}.\mathtt{AcceptReveal}(\tx_{\mathsf{reveal}})
    }
    \end{subprocedure}
    };
    \node (B) [fill=white, draw=black] at ([yshift=2cm] mess to){
    \begin{subprocedure}
      \procedure[linenumbering, bodylinesep=2pt]{$\mathtt{AcceptReveal}(\tx_{\mathsf{reveal}}):$}
      {
        \text{Require } \mathsf{GetStage}() = \mathsf{Reveal}\\ 
        \text{Parse } (\ek_i, \bal, \pi_{\mathsf{reveal}}) = \tx_{\mathsf{reveal}} \\
        \text{Get } c \text{ of } \ek_{i} \text{ from } \ccon \\
        \statement = (\bal, \ek_i, c)\\
        \text{Require: } \zkverify(\vk_{\mathsf{reveal}}, \statement)\\
        \text{Store: } (\ek_i, \bal)
      } 
    \end{subprocedure}
    };
    \draw [->] (A) -- node [midway,above] {$\tx_{\mathsf{reveal}}$} (B);
    \end{sequencediagram}
    }
    \vspace{-0.3cm}
    \caption{Unlinkable Weighted Voting Using~\aDDM: Protocol for Candidate. 
      }
    \label{fig:voting-candidate}
\end{figure*}

\section{Tagging Scheme Construction}
\paragraph{Tagging Scheme From Pseudorandom Function} 
We construct a tagging scheme from a One-way Pseudorandom function.  
\cref{fig:tagging-construction} provides the detailed construction. 

\begin{figure}[!h]
    \center
        \begin{pcvstack}[boxed]
            \procedure[space=keep,  bodylinesep=2pt]{$\tagsetup(\sparam)$}{
                F:\bset^{\lambda}\times\bset^{\lambda}\rightarrow \bset^{\lambda} \text{ be a pseudorandom function} 
            }
            \pcvspace
            \begin{pchstack}
            \procedure[space=keep,  bodylinesep=2pt]{$\tagkeygen(\dcmcoinsk)$}{
                \text{Parse: } (k,r) \leftarrow \dcmcoinsk\\
                \dcmcoinpk \leftarrow F(k, r)\\
                \pcreturn \dcmcoinpk
            }
            \procedure[space=keep,  bodylinesep=2pt]{$\tageval(\dcmcoinsk)$}{
                \text{Parse: } (k,r) \leftarrow \dcmcoinsk\\
                \nulli \leftarrow F(k, 0^\lambda)\\
                \pcreturn \nulli
            }
            \end{pchstack}
        \end{pcvstack}
    \caption{Construction of the Tagging Scheme used in DML.}
    \label{fig:tagging-construction}
\end{figure}

\begin{definition}[Security of Tagging Scheme]
\label{def:security-tagging-scheme}
A tagging scheme, \texttt{TAG},
    is said to be one-way if for any PPT adversary
    \adv, 
    $$\Pr[\tagOnewayExp(\sparam)] \leq \mathsf{negl}(\lambda)$$ 
    It is pseudorandom if, for any PPT adversary \adv, 
    $$|\Pr[\tagPRExp^0(\sparam)=1] - \Pr[\tagPRExp^1(\sparam)=1]|\leq \mathsf{negl}(\lambda)$$
\end{definition}

\begin{figure}[h]
    \centering
    \begin{pchstack}[boxed]
        \begin{pcvstack}
            \procedure[space=keep, bodylinesep=2pt]{$\tagOnewayExp(\sparam)$}{
                \tagpp \gets \tagsetup(\sparam)\\
                \dcmcoinsk \sample \mathcal{K} \\
                \dcmcoinpk \leftarrow \tagkeygen(\dcmcoinsk)     \\
                \nulli \leftarrow \tageval(\dcmcoinsk)   \\
                \dcmcoinsk^* \leftarrow \adv(\tagpp, \nulli, \dcmcoinpk)  \\
                \pcreturn \tageval(\dcmcoinsk^*) = \nulli
            }
        \end{pcvstack}
        \pchspace
        \begin{pcvstack}
            \procedure[space=keep,  bodylinesep=2pt]{$\tagPRExp^b(\sparam) $}{
                \tagpp \gets \tagsetup(\sparam)\\
                \dcmcoinsk \sample \mathcal{K} \\
                (\nulli,\dcmcoinpk) \leftarrow \fun{Ch}_{b,\dcmcoinsk}()   \\
                b'\leftarrow \adv(\tagpp, \nulli, \dcmcoinpk)  \\
                \pcreturn b'
            }

            \procedure[space=keep,  bodylinesep=2pt]{$\mathsf{Ch}_{b, \dcmcoinsk}() $}{
                \nulli \sample \mathcal{T}\\ 
                \dcmcoinpk \sample \mathcal{P}\\
                \pcif b = 0: \\
                \t\dcmcoinpk \define \tagkeygen(\dcmcoinsk)  \\
                \t\nulli \define \tageval(\dcmcoinsk)\\
                \pcreturn (\nulli, \dcmcoinpk)
            }
        \end{pcvstack}
    \end{pchstack}
    \caption{Definition of Oneway and Pseudorandom Experiments for the tagging scheme.}
    \label{fig:expdoublespending}
\end{figure}

In addition, we require the function $F$ to be \emph{collision resistant} on its entire inpput space, i.e., it is infeasible for any PPT adversary to find two $2\lambda$-bit strings $(k,r)$ and $(k',r')$ such that $F(k, r) = F(k', r')$. Looking ahead, our tagging scheme will be instantiated in the random oracle model, and hence the collision resistance required for our applications would immediately follow.


\section{Proofs}
\label{sec:missing-proofs}

\subsection{Proof of Theorem 1}
\repeattheorem{fixed}

We prove this theorem by proving this series of lemmas. 

\begin{lemma}
    \label{lemma:correctness-lemma1} Assume that \zksnark, $\Pi$, the underlying Merkle tree, $\mathtt{MT}$ are correct and the tagging scheme, $\mathtt{TAG}$ is deterministic, then $\fDDM$ is correct.
\end{lemma}
\begin{proof}
    It can be verified by inspection that the scheme $\fDDM$ is correct by relying on the correctness of the each of three ingredients, i.e., namely \zksnark, Merkle tree, and tagging scheme. 

\end{proof}

\begin{lemma}[\fDDM: No One-more Redeeming]
\label{thm:dcmsecurity-nodoubleredeem}
Assuming that the \zksnark, $\Pi$, is simulation extractable, $F$ for the tagging scheme and the hash function used in \tree~are modeled as the random oracle, then the construction in~\cref{fig:construction-fixed-amt} satisfies no double redeeming as defined in~\cref{def:dcm-double-spending}.
\end{lemma}

\begin{proof}
(\emph{Sketch}) Suppose there exists a PPT adversary $\adv$ such that $\Pr\left[ \dcmexpds_{\adv}(\sparam) \right]$ is non-negligible. This means that if we sample $\pp \gets \dcmsetup(\sparam)$ and provide $\adv$ with $\pp$, $\adv$ will output in PPT the lists $\{(\nulli_j, \dcmproof_j, m_j)\}_{j\in [n+1]}, \dcmncoins$ such that $\dcmverify(\dcmcoinaccstate, \nulli_j, \dcmproof_j, m_j)$ succeeds for all $j \in [n+1]$, where $\dcmcoinaccstate = \dcmaccumulatecoin(\dcmncoins)$, and for all $i \neq j \in [n+1]$, $\nulli_i \neq \nulli_j$. Since the proofs are simulation extractable, there is a PPT extractor we can run on the transcript of $\adv$ to obtain $\{(i_j,\dcmcoinsk_j,\adsproof_j)\}_{j\in [n+1]}$ such that\begin{equation*}
        \begin{aligned}
        & \tree.\mkverify(i_j, \dcmcoinpk_{i_j}^*, \dcmcoinaccstate, \adsproof_j)  \land  
        \\& \dcmcoinpk_{i_j}^* = \tagscheme.\tagkeygen(\dcmcoinsk_j) \land
        \\& \nulli_j = \tagscheme.\tageval(\dcmcoinsk_j)
        \end{aligned}
\end{equation*}
Since $\tree$ is a random oracle, with all but negligible probability, $\dcmcoinpk_{i_j}^* = \dcmcoinpk_{i_j}$ for all $j \in [n+1]$. Since for all $i \neq j \in [n+1]$, $\nulli_i \neq \nulli_j$, this means that for all $i \neq j \in [n+1]$, $\dcmcoinsk_i \neq \dcmcoinsk_j$. However, since $\dcmcoinpk_{i_j} = \tagscheme.\tagkeygen(\dcmcoinsk_j)$ for all $j \in [n+1]$, this means that we have been able to come up with some $\dcmcoinsk, \dcmcoinsk'$ such that $\tagkeygen(\dcmcoinsk)=\tagkeygen(\dcmcoinsk')$, as there are only at most $n$ $\dcmcoinpk_{i_j}$. This violates the collision resistance of the tagging scheme, giving us a contradiction. Therefore, the construction in the construction in~\cref{fig:construction-fixed-amt} satisfies no double redeeming as defined in~\cref{def:dcm-double-spending}.
\end{proof}

\begin{lemma}[\fDDM: Theft prevention]
\label{lemma:theft}
    Assume that zk-SNARK, $\Pi$, is simulation extractable, 
    $F$ for the tagging scheme and the hash function used in \tree~are modeled as the random oracle
    then the construction in~\cref{fig:construction-fixed-amt} provides theft prevention as defined in~\cref{def:dcm-theft}. 
\end{lemma}
\begin{proof}
(\emph{Sketch}) Suppose there exists a PPT adversary $\adv$ such that $\Pr\left[ \dcmexptheft_{\adv}(\sparam) \right]$ is non-negligible. This means that if we sample $\pp \gets \dcmsetup(\sparam)$ and provide $\adv$ with $\pp$ and oracle access to $\dcmcreatecoin\oracle$ and  $\dcmredeemcoin\oracle$, $\adv$ will output in PPT the list $\dcmncoins$ and $\nulli, \dcmproof, m$ such that $\dcmverify(\dcmcoinaccstate, \nulli, \dcmproof, m)$ succeeds, every $\dcmcoinpk \in \dcmncoins$ corresponds to a call to $\dcmcreatecoin\oracle$, and $\nulli, m$ do not correspond to a call to $\dcmredeemcoin\oracle$ with $\dcmcoinaccstate$, where $\dcmcoinaccstate = \dcmaccumulatecoin(\dcmncoins)$. Since the proof is simulation extractable, there is a PPT extractor we can run on the transcript of $\adv$ to obtain $(i,\dcmcoinsk,\adsproof)$ such that\begin{equation*}
        \begin{aligned}
        & \tree.\mkverify(i, \dcmcoinpk_{i}^*, \dcmcoinaccstate, \adsproof)  \land  
        \\& \dcmcoinpk_{i}^* = \tagscheme.\tagkeygen(\dcmcoinsk) \land
        \\& \nulli = \tagscheme.\tageval(\dcmcoinsk)
        \end{aligned}
\end{equation*}
Since $\tree$ is a random oracle, with all but negligible probability, $\dcmcoinpk_{i}^* = \dcmcoinpk_{i}$. Note that $\dcmcoinpk_{i}$ corresponds to a call to $\dcmcreatecoin\oracle$ and using $\adv$, we are able to extract a $\dcmcoinsk$ such that $\nulli = \tagscheme.\tageval(\dcmcoinsk)$. We now show how we can use $\adv$ to break the one-wayness of the tagging scheme. We receive $\tagpp, \nulli, \dcmcoinpk$ from the challenger for the one-wayness experiment for the tagging scheme. We sample the rest of $\pp$ honestly and provide $\pp$ to $\adv$. We then simulate the responses to the oracle calls made by $\adv$. We guess which of the oracle calls to $\dcmcreatecoin\oracle$ will eventually correspond to the response returned by $\adv$ and on that call, return the $\dcmcoinpk$ we received from the challenger for the one-wayness experiment for the tagging scheme. In all other instances, we honestly simulate $\dcmcreatecoin\oracle$. Next, we get to how we simulate $\dcmredeemcoin\oracle$. If $\dcmcoinpk_{i} \ne \dcmcoinpk$, we can honestly simulate $\dcmredeemcoin\oracle$. However, if $\dcmcoinpk_{i} = \dcmcoinpk$, since we do not have $\dcmcoinsk$, we do not have the witness to generate the proof. So, we use simulated proofs instead. Since $\nulli$ and $\dcmcoinpk$ are consistent, i.e., generated honestly by the challenger for the one-wayness experiment for the tagging scheme, the simulated proof is indistinguishable owing to the zero-knowledge property of the zk-SNARK. Since $\nulli, m$ do not correspond to a call to $\dcmredeemcoin\oracle$, the proof $\dcmproof$ returned by $\adv$ will not be one of the simulated proofs we provided with all but non-negligible probability. Hence, our extraction will succeed. Finally, we simply forward back the extracted $\dcmcoinsk$. Since $\adv$ is PPT and succeeds with non-negligible probability, so do we. This violates the one-wayness of the tagging scheme, giving us a contradiction. Therefore, the construction in the construction in~\cref{fig:construction-fixed-amt} provides theft prevention as defined in~\cref{def:dcm-theft}.
\end{proof}

\begin{lemma}[\fDDM: Non-slanderability]\label{thm:dcmsecurity-nonslander}
    Assume that the function $F$ for the tagging scheme and the hash function used in \tree~are modeled as the random oracle
    \zksnark, $\Pi$, is simulation extractable, 
    then $\fDDM$ satisfies the non-slanderability property defined in \cref{def:dcm-nonslanderability}.
\end{lemma}
\begin{proof}
(\emph{Sketch}) Suppose there exists a PPT adversary $\adv$ such that $\Pr\left[ \dcmnonslander_{\adv}(\sparam) \right]$ is non-negligible. This means that if we sample $\pp \gets \dcmsetup(\sparam)$ and provide $\adv$ with $\pp$ and oracle access to $\dcmcreatecoin\oracle$ and  $\dcmredeemcoin\oracle$, $\adv$ will output in PPT $\dcmcoinpk, \dcmcoinpk^*$ and $\nulli^*, \dcmproof^*, m^*$ such that $\dcmcoinpk$ corresponds to a call to $\dcmcreatecoin\oracle$, $\nulli^* = \nulli$ where $(\nulli, \pi)\leftarrow\dcmredeemcoin\oracle((\dcmcoinpk, \dcmcoinpk^*), 1, m)$, $\dcmverify(\dcmcoinaccstate, \nulli^*, \dcmproof^*, m^*)$ succeeds, and $\nulli^*, m^*$ do not correspond to a call to $\dcmredeemcoin\oracle$ with $\dcmcoinaccstate$, where $\dcmcoinaccstate = \dcmaccumulatecoin(\dcmcoinpk, \dcmcoinpk^*)$. Since the proof is simulation extractable, there is a PPT extractor we can run on the transcript of $\adv$ to obtain $(i,\dcmcoinsk,\adsproof)$ such that\begin{equation*}
        \begin{aligned}
        & \tree.\mkverify(i, \dcmcoinpk_{i}^*, \dcmcoinaccstate, \adsproof)  \land  
        \\& \dcmcoinpk_{i}^* = \tagscheme.\tagkeygen(\dcmcoinsk) \land
        \\& \nulli^* = \tagscheme.\tageval(\dcmcoinsk)
        \end{aligned}
\end{equation*}
Since $\tree$ is a random oracle, with all but negligible probability, $\dcmcoinpk_{i}^* = \dcmcoinpk$ or $\dcmcoinpk_{i}^* = \dcmcoinpk^*$. Note that $\dcmcoinpk$ corresponds to a call to $\dcmcreatecoin\oracle$ and using $\adv$, we are able to extract a $\dcmcoinsk$ such that $\nulli = \nulli^* = \tagscheme.\tageval(\dcmcoinsk)$. We now show how we can use $\adv$ to break the one-wayness of the tagging scheme. We receive $\tagpp, \nulli, \dcmcoinpk$ from the challenger for the one-wayness experiment for the tagging scheme. We sample the rest of $\pp$ honestly and provide $\pp$ to $\adv$. We then simulate the responses to the oracle calls made by $\adv$. We guess which of the oracle calls to $\dcmcreatecoin\oracle$ will eventually correspond to the response returned by $\adv$ and on that call, return the $\dcmcoinpk$ we received from the challenger for the one-wayness experiment for the tagging scheme. In all other instances, we honestly simulate $\dcmcreatecoin\oracle$. Next, we get to how we simulate $\dcmredeemcoin\oracle$. If $\dcmcoinpk_{i} \ne \dcmcoinpk$, we can honestly simulate $\dcmredeemcoin\oracle$. However, if $\dcmcoinpk_{i} = \dcmcoinpk$, since we do not have $\dcmcoinsk$, we do not have the witness to generate the proof. So, we use simulated proofs instead. Since $\nulli$ and $\dcmcoinpk$ are consistent, i.e., generated honestly by the challenger for the one-wayness experiment for the tagging scheme, the simulated proof is indistinguishable owing to the zero-knowledge property of the zk-SNARK. Since $\nulli^*, m^*$ do not correspond to a call to $\dcmredeemcoin\oracle$ and $\nulli^* = \nulli$, the proof $\dcmproof$ returned by $\adv$ will not be one of the simulated proofs we provided with all but non-negligible probability. Hence, our extraction will succeed. Finally, we simply forward back the extracted $\dcmcoinsk$. Since $\adv$ is PPT and succeeds with non-negligible probability, so do we. This violates the one-wayness of the tagging scheme, giving us a contradiction. Therefore, the construction in the construction in~\cref{fig:construction-fixed-amt} satisfies the non-slanderability property as defined in~\cref{def:dcm-nonslanderability}.
\end{proof}

\begin{lemma}[\fDDM: Unlinkability]\label{thm:dcmsecurity-unlink}
    Assume that
    the function $F$ in the tagging scheme is modeled as a random oracle,  
    \zksnark, $\Pi$, is simulation extractable and zero-knowledge, and 
    \tree is instantiated with a random oracle, 
    then construction in~\cref{fig:construction-fixed-amt}  provides unlinkability as defined in~\cref{def:dcm-unlinkability}.
\end{lemma}
\begin{proof}
(\emph{Sketch}) We show that for any PPT adversary $\adv$, $\Pr\left[ \dcmexpunlink_{\adv}(\sparam) \right]$ is negligible via a series of hybrids. In our first hybrid, $H_0$, we deterministically set $b = 0$ in $\dcmexpunlink_{\adv}(\sparam)$, and in our last hybrid, $H_5$, we deterministically set $b = 1$ in $\dcmexpunlink_{\adv}(\sparam)$. We will show that $H_0 \approx H_5$. In hybrid $H_1$, we switch from providing honestly generated proofs to providing simulated proofs. Since the tags are consistent, the simulated proofs are indistinguishable owing to the zero-knowledge property of the zk-SNARK, and hence $H_0 \approx H_1$. Next, in $H_2$, we switch the tags with random strings. Owing to the pseudorandomness of the tagging scheme, $H_1 \approx H_2$. Since none of the output now depends on $b$, in $H_3$, we switch $b=0$ to $b=1$ and $H_2 \equiv H_3$. Next, in $H_4$, we switch the tags back from random strings to honestly generated tags. Owing to the pseudorandomness of the tagging scheme, $H_3 \approx H_4$. Finally, in $H_5$, we switch the simulated proofs back to honestly generated proofs. ince the tags are consistent, the simulated proofs are indistinguishable owing to the zero-knowledge property of the zk-SNARK, and hence $H_4 \approx H_5$. This completes the hybrid argument and the proof of the lemma. Therefore, the construction in the construction in~\cref{fig:construction-fixed-amt} provides unlinkability as defined in~\cref{def:dcm-unlinkability}. 
\end{proof}

\subsection{Proofs for Theorem 2}
\repeattheorem{arbit}
We prove this theorem by proving the following lemmas.
\begin{lemma}
\label{lemma:correctness-lemma1} 
Assume that \zksnark, $\Pi$, the underlying Merkle tree, $\tree$, $\mathtt{E}$, $\mathtt{Com}$ are correct and the tagging scheme, $\tagscheme$ is deterministic, then $\aDDM$ is correct.
\end{lemma}
\begin{proof}
 It can be verified by inspection that the scheme $\aDDM$ is correct by relying on the correctness of the each of underlying primitives, i.e., namely \zksnark, Merkle tree, tagging scheme, encryption, and commitment scheme.     
\end{proof}

\begin{lemma}[\aDDM: No One-more Redeeming]
\label{thm:addm-nodoubleredeem}
Assuming that the \zksnark, $\Pi$, is simulation extractable, $F$ for tagging scheme and the hash function used in \tree~are modeled as the random oracle, then the construction in~\cref{fig:construction-arbitrary-amt} satisfies no double redeeming as defined in~\cref{def:dcm-double-spending}.
\end{lemma}

\begin{proof}
The proof of this lemma is similar to that of Lemma \ref{thm:dcmsecurity-nodoubleredeem} with the only difference being the binding of the commitment scheme. 
\end{proof}

\begin{lemma}[\aDDM: Theft prevention]
    Assume that zkSNARK, $\Pi$, is simulation extractable, 
    and the tagging scheme, \tagscheme, is one-way, 
    then the construction in~\cref{fig:construction-arbitrary-amt} provides theft prevention as defined in~\cref{def:dcm-theft}. 
\end{lemma}
\begin{proof}
The proof of this lemma is similar to that of Lemma \ref{lemma:theft}. 
\end{proof}
\begin{lemma}[\aDDM: Non-slanderability]
    Assume that the function $F$ for the tagging scheme and the hash function used in \tree~are modeled as the random oracle,
    \zksnark, $\Pi$, is simulation extractable, 
    then $\aDDM$ satisfies the non-slanderability property defined in \cref{def:dcm-nonslanderability}.
\end{lemma}
\begin{proof}
The proof of this lemma is similar to that of Lemma \ref{thm:dcmsecurity-nonslander}. 
\end{proof}

\begin{lemma}[\fDDM: Unlinkability]
    Assume that the function $F$ for the tagging scheme and the hash function used in \tree~are modeled as the random oracle,
    \zksnark, $\Pi$, is simulation extractable and zero-knowledge, 
    $\mathsf{Com}$ is hiding and biding,  
    $\mathtt{E}$ is IND-CPA, 
    then the construction in~\cref{fig:construction-arbitrary-amt}  provides unlinkability as defined in~\cref{def:dcm-unlinkability}.
\end{lemma}
\begin{proof}
(\emph{Sketch}) 
The main difference between the two constructions is the use of commitment and encryption schemes. 
The hiding property of $\commit$ ensures that the adversary gains no information from the challenged public keys, thereby rendering it unable to distinguish which public key corresponds to which data.

We show that for any PPT adversary $\adv$, $\Pr\left[ \dcmexpunlink_{\adv}(\sparam) \right]$ is negligible via a series of hybrids. 
In our first hybrid, $H_0$, we deterministically set $b = 0$ in $\dcmexpunlink_{\adv}(\sparam)$, and in our last hybrid, $H_5$, we deterministically set $b = 1$ in $\dcmexpunlink_{\adv}(\sparam)$. 
We will show that $H_0 \approx H_5$. 
In hybrid $H_1$, we switch from providing honestly generated proof and ciphertext pair to providing simulated proof and an encryption of $0$. Since the tags are consistent, the simulated proofs are indistinguishable owing to the zero-knowledge property of the zk-SNARK and the encryption of $0$ is indistinguishable owing to the IND-CPA security of the encryption scheme, and hence $H_0 \approx H_1$. Next, in $H_2$, we switch the tags with random strings. Owing to the pseudo-randomness of the tagging scheme, $H_1 \approx H_2$. Since none of the output now depends on $b$, in $H_3$, we switch $b=0$ to $b=1$ and $H_2 \equiv H_3$. Next, in $H_4$, we switch the tags back from random strings to honestly generated tags. Owing to the pseudorandomness of the tagging scheme, $H_3 \approx H_4$. Finally, in $H_5$, we switch the simulated proof and the random string back to honestly generated proof and the ciphertext. Since the tags are consistent, the simulated proofs are indistinguishable owing to the zero-knowledge property of the zk-SNARK and the ciphertext is indistinguishable owing to the IND-CPA security of the encryption scheme, and hence $H_4 \approx H_5$. This completes the hybrid argument and the proof of the lemma. Therefore, the construction in the construction in~\cref{fig:construction-arbitrary-amt} provides unlinkability as defined in~\cref{def:dcm-unlinkability}. 
\end{proof}

\pparagraph{On the Non-Blackbox Use of Cryptographic Primitives and Random Oracles} We remark that while a non-black-box use of a hash function that is modeled as a random oracle within a cryptographic protocol is not considered standard in theoretical cryptography (see Section 1.4 of~\cite{TCC:AgrKopWat18} for one example) particularly considering impossibility or lower bound results in the literature~\cite{STOC:ImpRud89}, it is \emph{not} uncommon to utilize concrete and practical hash functions (e.g., SHA256) in a non-black-box manner (say in a zero-knowledge cryptographic protocol). For instance, towards achieving a practical solution while at the same time having a theoretical justification, many prior works relied on the non-black-box usage of practical hash functions, e.g.,~\cite{CCS:AKSY22} and \cite{sasson2014zerocash}. We refer the reader to~\cite{STOC:GOSV14} for a discussion studying the gap between black-box and non-black-box constructions, particularly in the contexts of zero-knowledge and MPC protocols.